\documentclass[a4paper,12pt]{article}
\usepackage{amsmath,amssymb,amsfonts,bm,xcolor}
\usepackage{algorithm,algorithmic,url}
\usepackage{enumitem,hyperref,natbib,multirow,rotating,sectsty,booktabs}

\newcommand{\algrule}[1][.2pt]{\par\vskip.5\baselineskip\hrule height #1\par\vskip.5\baselineskip}

\newcommand{\dn}{\textrm{d}}       
\newcommand{\dd}{\, \textrm{d}}       
\newcommand{\dint}{\int \!}       

\newtheorem{lemma}{Lemma}

\newtheorem{proposition}{Proposition}
\newtheorem{proof}{Proof}

\title{\bf Correlated pseudo-marginal Metropolis-Hastings\\ using quasi-Newton proposals}

\author{Johan Dahlin, Adrian~Wills and Brett~Ninness%
\thanks{E-mail adresses to authors: JD (corresponding author) \url{uni@johandahlin.com}, AW \url{adrian.wills@newcastle.edu.au} and BN \url{brett.ninness@newcastle.edu.au}. JD and AW are with the School of Engineering, The University of Newcastle, Callaghan, NSW 2308, Australia. BN is with the Faculty of Engineering and Built Environment, The University of Newcastle, Callaghan, NSW 2308, Australia.}%
}

\begin{document}
\maketitle

\noindent
\textbf{Abstract}

\noindent
Pseudo-marginal Metropolis-Hastings (pmMH) is a versatile algorithm for sampling from target distributions which are not easy to evaluate point-wise.
However, pmMH requires good proposal distributions to sample efficiently from the target, which can be problematic to construct in practice.
This is especially a problem for high-dimensional targets when the standard random-walk proposal is inefficient.
    \\ \\
We extend pmMH to allow for constructing the proposal based on information from multiple past iterations.
As a consequence, quasi-Newton (qN) methods can be employed to form proposals which utilize gradient information to guide the Markov chain to areas of high probability and to construct approximations of the local curvature to scale step sizes.
The proposed method is demonstrated on several problems which indicate that qN proposals can perform better than other common Hessian-based proposals. \\

\noindent \textbf{Keywords}

\noindent
Markov chain Monte Carlo, Bayesian inference, Riemann manifolds, Hessian estimation, Quasi-Newton methods, Machine learning.

\newpage


\section{Introduction}
\label{sec:introduction}
Sampling from some target distribution $\pi(\theta)$ with $\theta \in \Theta \subset \mathbb{R}$ is a common problem in statistics and Machine learning.
This problem arises when exploring the posterior in the Bayesian statistical paradigm \citep{Robert2007,GelmanCarlinSternDunsonVehtariRubin2013} or the cost function in a supervised learning problem \citep{Murphy2012,Ghahramani2015,WillsSchon2018}.

A standard choice is to make use of Metropolis-Hastings (MH; \citealp{RobertCasella2004}) to sample from $\pi(\theta)$.
This is done by constructing Markov chain which has the sought target distribution as its stationary distributions.
Samples from $\pi(\theta)$ can therefore be generated by simulating this chain.

However, this requires point-wise evaluation of $\pi(\theta)$ which can be intractable or computationally prohibitive.
The former can be the result of the model containing missing data or latent variables.
The latter is common in intractable likelihood models \citep{MarinPudloRobertRyder2011} and data-rich scenarios \citep{BardenetDoucetHolmes2017}.

In some situations, it is possible to circumvent this problem by changing the target into something that can be (efficiently) evaluated point-wise.
MH can then be used to sample from this so-called \textit{extended target},
\begin{align}
	\bar{\pi}(\theta, u)
	=
	\widehat{\pi}(\theta|u)
	\,
	m_{\theta}(u),
	\label{eq:extendedtarget}
\end{align}
where $\widehat{\pi}(\theta|u)$ denotes some unbiased estimator of the original target $\pi(\theta)$.
Here, we introduce some auxiliary variables denoted $u \in \mathcal{U}$ with density $m_{\theta}(u)$ to facilitate point-wise evaluation of the intractable target.
This is the core idea of the pseudo-marginal MH (pmMH; \citealp{AndrieuRoberts2009}) algorithm, which can be seen as an \emph{exact approximation} of the intractable ideal MH algorithm targeting $\pi(\theta)$.
That is, pmMH makes use of noisy evaluations from the target but still generates samples from the original target in the same manner as the exact ideal algorithm would.

A concrete example of a pmMH algorithm is when sequential Monte Carlo (SMC; \citealp{DelMoralDoucetJasra2006}) is employed to construct $\widehat{\pi}(\theta|u)$ as direct point-wise evaluation of $\pi(\theta)$ is not possible.
This leads to the particle MH (pMH; \citealp{AndrieuDoucetHolenstein2010}) algorithm, which allows for Bayesian inference in state-space models (SSMs).
Another example considered by \cite{WillsSchon2018} is to learn the weights in e.g., deep neural networks \citep{LeCunBengioHinton2015}, where direct evaluation of $\pi(\theta)$ is computationally prohibitive due to the data size.

The performance of pmMH depends on: (i) good parameter proposals and (ii) computationally efficient estimators $\widehat{\pi}(\theta|u)$.
These are important choices to research to promote a wider adoption of pmMH.

The main contribution of this paper is an extension of pmMH to allow the proposal to make use of information from multiple previous states.
Usually, the proposal can only depend on the last state of the Markov chain but this  requirement is relaxed by extending earlier work for MH \citep{RobertsGilks1994,ZhangSutton2011}.
This enables making use of quasi-Newton (qN; \citealp{NocedalWright2006}) methods from optimisation to construct efficient proposals.
Specifically, our contribution entails:
\begin{itemize}
	\item[-] two novel proposals based on symmetric rank-one (SR1) method with a novel trust-region approach and a tailored regularized least squares (LS) method,
	\item[-] a proof of the validity of pmMH with memory,
	\item[-] extensive numerical benchmarks showing the benefits of the novel proposals over standard choices.
\end{itemize}

The numerical benchmarks indicate the qN proposals can be used to compute accurate estimates of the Hessian with only a minor computational overhead.
As a result, the novel proposals can out-perform direct computations of the Hessian or the use of SMC methods in terms of computational speed and/or quality of the posterior estimates.
Hence, qN proposals can be a useful alternative to other Hessian-based proposals.

Most popular proposals in pmMH are currently based on discretisations of Langevin diffusions \citep{RobertsRosenthal1998,GirolamiCalderhead2011} which can be extended to pmMH as discussed by \cite{DahlinLindstenSchon2015a} and \cite{NemethSherlockFearnhead2016}.
These proposals require accurate estimates of the gradient and Hessian of the log-target to perform well.
In practice, this can be a problem as the Hessian can be difficult to estimate accurately which results in a computational bottleneck.

One approach to circumvent this problem is to employ qN methods to construct a local Hessian approximation from gradient information which is typically faster and easier to estimate accurately.
This idea has been employed for MH \citep{ZhangSutton2011}, pMH \citep{DahlinLindstenSchon2015b,DahlinWillsNinness2018b} and Hamiltonian Monte Carlo (HMC; \citealp{SimsekliBadeauCemgilRichard2016}).
This work extends the use of qN to pmMH, proves the validity of this idea and addresses some practical implementation issues.

The idea of adapting the proposal on-the-fly within MH and pmMH is common and has received a lot of attention by e.g., \cite{AndrieuThoms2008}.
A similar idea to the one presented in this paper is discuss for MH by e.g., \cite{GilksRobertsGeorge1994}, \cite{HaarioSaksmanTamminen2001} and \cite{CaiMeyerPerron2008}.
However, a major difference is that no gradient and local Hessian information is used when constructing the proposal.

\section{An overview of the proposed algorithm}
\label{sec:pmmh}
This section aims to introduce the pmMH algorithm with memory. We begin by briefly discussing standard pmMH to set the notation and then continue with presenting the extension which allows for making use of information from multiple previous iterations to construct good proposals.
General introductions to pmMH and pMH are given by \cite{AndrieuRoberts2009}, \cite{AndrieuDoucetHolenstein2010} and \cite{DahlinSchon2017}.

\subsection{Standard pseudo-marginal MH}
\label{sec:pmmh:standard}
Consider the extended target density in \eqref{eq:extendedtarget} on the space $\Theta \times \mathcal{U}$, where $\Theta$ denotes the space of interest with the target density $\pi(\theta)$.
As discussed in the introduction, it is sometimes not possible to implement MH directly to sample from $\pi(\theta)$ as point-wise evaluation of the density is not possible e.g., when the model contains latent, missing or unknown variables.

Instead, we introduce the auxiliary variables $u$ and form the estimator $\widehat{\pi}(\theta | u)$.
If this estimator is unbiased, then a MH algorithm targeting $\bar{\pi}(\theta, u)$ can generate samples from $\pi(\theta)$ by marginalisation.
To see why this scheme is valid, note that the unbiasedness of $\widehat{\pi}(\theta | u)$ can be expressed as
\begin{align}
	\pi(\theta)
	=
	\dint
	\bar{\pi}(\theta, u)
	\dd u
	=
	\dint
	\widehat{\pi}(\theta | u)
	\,
	m_{\theta}(u)
	\dd u.
	\label{eq:pmmh:unbiasedness}
\end{align}
This means that we can recover the target density by marginalisation of the extended target density.
Moreover, consider the MH algorithm targeting the extended target, which consists of two steps: (i) sampling a candidate from the proposal and then (ii) accepting/rejecting the candidate.
In the first step, the candidate is sampled from a proposal distribution,
\begin{align}
	q(\theta', u' | \theta_{k-1}, u_{k-1})
	=
	q(\theta' | \theta_{k-1}, u_{k-1})
	q(u' | u_{k-1}),
	\label{eq:pmmh:proposal}
\end{align}
where $q(\theta' | \cdot)$ and $q(u' | \cdot)$ are specified by the user and these choices are discussed in more detail in Section~\ref{sec:proposals}.
In the second step, the proposed parameters $\{\theta', u'\}$ are accepted or rejected with the probability,
\begin{align}
	\alpha_k
	&=
	1
	\wedge
	\frac{\bar{\pi}(\theta', u')}{\bar{\pi}(\theta_{k-1}, u_{k-1})}
	\frac{q(\theta_{k-1}, u_{k-1} | \theta', u')}{q(\theta', u' | \theta_{k-1}, u_{k-1})}
	\label{eq:pmmh:aprob}
	\\
	&=
	1
	\wedge
	\frac{\widehat{\pi}(\theta' | u')}{\widehat{\pi}(\theta_{k-1} | u_{k-1})}
	\frac{m_{\theta'}(u')}{m_{\theta_{k-1}}(u_{k-1})}
	\frac{q(\theta_{k-1}, u_{k-1} | \theta', u')}{q(\theta', u' | \theta_{k-1}, u_{k-1})},
	\nonumber
\end{align}
where the expanded expression is given to simplify the exposition in the subsequent sections and with $a \wedge b = \min(a,b)$.
Hence, we can run a standard MH algorithm targeting $\bar{\pi}(\theta, u)$ to obtain $\{\theta_k, u_k\}_{k=1}^K$.

We obtain pMH as a special case of pmMH when the target is the parameter posterior distribution,
\begin{align}
	\pi(\theta)
	\triangleq
	\frac{p(\theta) p(y|\theta)}{p(y)}
	=
	\frac
	{p(\theta) p(y|\theta)}
	{\dint p(\theta') p(y | \theta') \dd \theta'},
	\label{eq:parameter:posterior}
\end{align}
where $p(\theta)$, $p(y|\theta)$ and $p(y)$ denote the parameter prior, the likelihood and the marginal likelihood, respectively.
An SMC algorithm with $N$ particles is also used to construct $\widehat{p}^N(y | u, \theta)$ such that
\begin{align}
	\widehat{\pi}(\theta | u)
	=
	p(\theta)
	\widehat{p}^N(y | u, \theta).
	\label{eq:pmh:posterior:estimator}
\end{align}
The auxiliary variables $u \in \mathcal{U}$ are in this case all the random variables generated during a run of SMC or equivalently the particles and their ancestry lineage.
That is, $u$ contains all the information to reconstruct the estimator $\widehat{p}^N(y | u, \theta)$ and the SMC algorithm is deterministic given $u$.
We return to pMH in Section~\ref{sec:results:sv}.

\subsection{Allowing pmMH to use a memory}
\label{sec:pmmh:memory}
We can introduce memory into pmMH by forming the $M$-fold product of the extended target \eqref{eq:extendedtarget}.
This gives the \textit{extended product target} on the space $\Theta^M \times \mathcal{U}^M$ with the density given by
\begin{align}
	\bar{\pi}^M(\vartheta)
	=
	\prod_{i=1}^M
	\bar{\pi}(\bar{\theta}_i, \bar{u}_i)
	=
	\prod_{i=1}^M
	\widehat{\pi}(\bar{\theta}_i | \bar{u}_i)
	m_{\bar{\theta}_i}(\bar{u}_i),
	\label{eq:extendedtarget:memory}
\end{align}
with $\vartheta = \{\bar{\theta}, \bar{u}\}$, $\bar{\theta} = (\bar{\theta}_1, \ldots, \bar{\theta}_M)$ and $\bar{u} = (\bar{u}_1, \ldots, \bar{u}_M)$.

The product target \eqref{eq:extendedtarget:memory} consists of $M$ copies of the original target \eqref{eq:extendedtarget} referred to as sub-targets.
Hence, each component of $\bar{\theta}$ corresponds to one complete set of parameters for \eqref{eq:extendedtarget}.
This is the reason for the introduction of the bar symbol to distinguish between the $i$th component of the vector $\theta$ in \eqref{eq:extendedtarget} and the $i$th component of $\bar{\theta}$, which is in itself a parameter vector.
Finally, note that one sample from $\bar{\pi}^M(\vartheta)$ therefore consists of $M$ realisations of possible parameter vectors for \eqref{eq:extendedtarget}.

As usual, we can recover the target distribution in terms of $\bar{\theta}$ by marginalisation if the estimator $\widehat{\pi}(\bar{\theta}_i | \bar{u}_i)$ is unbiased.
That is,
\begin{align*}
	\pi(\bar{\theta})
	=
	\dint
	\bar{\pi}^M(\vartheta)
	\dd \bar{u}
	=
	\prod_{i=1}^M
	\dint
	\widehat{\pi}(\bar{\theta}_i | \bar{u}_i)
	m_{\bar{\theta}_i}(\bar{u}_i)
	\dd \bar{u}_i
	=
	\prod_{i=1}^M
	\pi(\bar{\theta}_i)
	,
\end{align*}
where each sub-target is assumed to be independent and by using the unbiasedness of $\widehat{\pi}(\bar{\theta}_i | \bar{u}_i)$ in \eqref{eq:pmmh:unbiasedness}.

We can now sample from \eqref{eq:extendedtarget:memory} by using a Metropolis-within-Gibbs \citep{RobertCasella2004, RobertsRosenthal2006} scheme as discussed in \cite{ZhangSutton2011}.
Hence, a sample from the $i$th component of $\vartheta$ is generated at iteration $k$, $\vartheta_{k,i} = \{\bar{\theta}_{k,i}, \bar{u}_{k,i}\}$, while keeping the other components fixed.
This amounts to sampling $\vartheta'_i$ given $\vartheta_i$ from a Markov kernel denoted $R$ by
\begin{align}
	\vartheta'_i
	\sim
	R
	\left(
		\vartheta'_i
		|
		\vartheta_i,
		\vartheta_{\setminus i}
	\right),
	\qquad
	i \in \{1, \ldots, M\},
	\label{eq:mh:within:gibbs:update}
\end{align}
where the dependence of $k$ is suppressed for brevity.
The Markov kernel corresponds to one iteration of a pmMH algorithm and the exact form of the Markov kernel is given later in Section~\ref{sec:theory}.
To simplify the notation, we introduce the \textit{memory of the Markov chain},
\begin{align}
	\vartheta_{\setminus i}
	=
	\{
		\vartheta'_{1:i-1},
		\vartheta_{i+1:M}
	\},
	\label{eq:pmmh:memory:def}
\end{align}
which is the information available within the memory for component $i$ during the iteration.
The dependence on the iteration number $k$ is suppressed for brevity.
The memory $\vartheta_{\setminus i}$ consists of the $i-1$ elements which already have been updated during the current iteration and the $M-i-2$ remaining elements.

Repeated sampling from \eqref{eq:mh:within:gibbs:update} will ultimately produce $K$ samples from the extended product target $\bar{\pi}^M(\vartheta)$ denoted by $\{\vartheta_k\}_{k=1}^K$.
The proposal for each component $i$ during the iteration $k$ can make use of the information in $\vartheta_{k, \setminus i}$.
Hence, we can write the proposal as
\begin{align}
	q(
		\vartheta_i'
		|
		\vartheta_{i},
		\vartheta_{\setminus i}
	)
	=
	q
	\left(
		\bar{\theta}'_i
		|
		\bar{u}'_i,
		\bar{\theta}_{i},
		\vartheta_{\setminus i}
	\right)
	q(
		\bar{u}'_i
		|
		\bar{u}_i
	).
	\label{eq:theory:proposal}
\end{align}
The resulting acceptance probability is given by
\begin{align}
	\alpha(\vartheta'_i, \vartheta_i)
	=
	1
	\wedge
	\frac
	{\bar{\pi}^M(\vartheta_i', \vartheta_{\setminus i})}
	{\bar{\pi}^M(\vartheta_i, \vartheta_{\setminus i})}
	\frac
	{q(\vartheta_i | \vartheta_i', \vartheta_{\setminus i})}
	{q(\vartheta_i' | \vartheta_i, \vartheta_{\setminus i})}
	=
	1
	\wedge
	\frac
	{\bar{\pi}(\vartheta_i')}
	{\bar{\pi}(\vartheta_i)}
	\frac
	{q(\vartheta_i | \vartheta_i', \vartheta_{\setminus i})}
	{q(\vartheta_i' | \vartheta_i, \vartheta_{\setminus i})},
	\label{eq:theory:aprob}
\end{align}
which follows from the sub-targets related to $\vartheta_{\setminus i}$ cancelling as they are fixed.
The complete procedure is given in Algorithm~\ref{alg:pmmh:memory}.

A systematic scan is used in this paper where each Markov kernel is applied in sequence to update the chain.
It is however possible to also make use of random scan, which applies the kernels in a random sequence.
This might further increase the performance but is left as future work.

\begin{algorithm}[!t]
    \caption{\textsf{Pseduo-marginal MH with memory}}
    \footnotesize
	\textsc{Inputs:} $K>0$, $M>1$, $\bar{\theta}^{(0)}$ and $q(\vartheta_i'|\vartheta_{k-1,i},\vartheta_{k, \setminus i})$.\\
	\textsc{Output:} $\{\vartheta^{(k)}\}_{k=1}^K$.
	\algrule[.4pt]
	\begin{algorithmic}[1]
		\STATE Sample $\bar{u}^{(0)}_i$ from $m(\bar{u}^{(0)}_i)$ for $i \in \{1, \ldots, M\}$.
		\STATE Compute $\widehat{\pi}(\bar{\theta}^{(0)}_i | \bar{u}^{(0)}_i)$ for $i \in \{1, \ldots, M\}$.
		\FOR{$k=1$ to $K$ and $i=1$ to $M$}
			\STATE Sample candidate parameter and auxiliary variables by
			$\vartheta'_{k,i} \sim q( \vartheta_{k, i}' | \vartheta_{k-1, i}, \vartheta_{k, \setminus i})$
			using the proposal \eqref{eq:theory:proposal}.
			\STATE Sample $\omega_i^{(k)}$ uniformly over $[0,1]$.
			\IF{$\omega_i^{(k)} \leq \alpha(\vartheta'_{k,i}, \vartheta_{k-1,i})$ given by \eqref{eq:theory:aprob}.}
				\STATE Accept $\vartheta'_{k,i}$ by $\vartheta_{k,i} \leftarrow \vartheta'_{k,i}$.
			\ELSE
				\STATE
				Reject $\vartheta'_{k,i}$ by $\vartheta_{k,i} \leftarrow \vartheta_{k-1,i}$.
			\ENDIF
		\ENDFOR
	\end{algorithmic}
	\label{alg:pmmh:memory}
\end{algorithm}

\section{Designing good proposals}
\label{sec:proposals}
The design of proposal distributions is instrumental for obtaining good performance in pmMH.
Recall, that two different proposal distributions are required for Algorithm~\ref{alg:pmmh:memory}, i.e., one for $u$ and another for $\theta$.
We discuss these in turn in the context of standard pmMH without a memory to ease the notation.
We consider the general case with memory in Section~\ref{sec:hessian:impdetails}.

\subsection{Proposal for auxiliary variables}
\label{sec:proposals:auxiliary}
The standard choice for the proposal for the auxiliary variables is to sample from the prior, i.e.,
\begin{align*}
	q(u' | \theta_{k-1}, u_{k-1})
	=
	m_{\theta_{k-1}}(u')
	=
	m(u'),
\end{align*}
which is known as the independent proposal as $m(u)$ is assumed not to be dependent on $\theta$.
The expression for $m(u)$ depends on the form of the estimator $\widehat{\pi}^N(\theta | u)$ but a common choice is a standard Gaussian distribution with mean zero and a identity matrix as its covariance.
This results in the acceptance probability \eqref{eq:pmmh:aprob} simplifying to
\begin{align}
	\alpha_k
	&=
	1
	\wedge
	\frac{\widehat{\pi}^N(\theta' | u')}{\widehat{\pi}^N(\theta_{k-1} | u_{k-1})}
	\frac{q(\theta_{k-1} | \theta', u')}{q(\theta' | \theta_{k-1}, u_{k-1})},
	\label{eq:pmmh:aprob:simplified}
\end{align}
by using the structure of the proposal in \eqref{eq:pmmh:proposal}.

Another alternative is the Crank-Nicholson (CN; \citealp{DahlinLindstenKronanderSchon2015,DeligiannidisDoucetPitt2017}) proposal which essentially corresponds to a first-order auto-regressive process given by
\begin{align}
	q(u' | \theta_{k-1}, u_{k-1})
	=
	\mathcal{N}
	\left(
		u';
		\sqrt{1 - \sigma^2_u}
		\,
		u_{k-1},
		\sigma^2_u
		I_{n_u}
	\right),
	\label{eq:pmmh:cn:proposal}
\end{align}
for some step length $\sigma_u > 0$ and with $n_u = |u|$, i.e., the number of elements in $u$. Moreover, we select standard Gaussian distribution as $m_{\theta}(u)$, which results in that the acceptance probability is again given by \eqref{eq:pmmh:aprob:simplified}.

This construction is only possible if $u$ can be transformed into any random variable required by the estimator $\widehat{\pi}^N(\theta | u)$.
Often this is not a problem as inverse CDF transformations can be used to convert the Gaussian random number into random numbers from most standard distributions.
The independent proposal for $u$ is recovered when $\sigma_u=1$.
Selecting $\sigma_u=0$ corresponds to fixed auxiliary variables, which would result in biased estimates of the target.

The correlation introduced by \eqref{eq:pmmh:cn:proposal} can have a beneficial impact on the performance of pmMH.
This is results in the reduction of the probability of the Markov chain getting stuck in a certain state for many iterations.
This can occur if $\widehat{\pi}^N(\theta' | u') \gg \pi(\theta')$ leading to the candidate parameter being accepted.
After this, it is difficult to get any candidate accepted as the estimate of the log-target at the current state is much larger than its true value.
The correlation in $u$ leads to that the estimate of the log-target for the subsequent candidates are also larger than their true values and therefore increases the probability of accepting the candidate and for the Markov chain to get loose.

\subsection{Proposal for parameters}
\label{sec:proposals:parameters}
A popular choice for the parameter proposal is a Gaussian distribution,
\begin{align}
	q(\theta' | \theta_{k-1}, u_{k-1})
	=
	\mathcal{N}
	\left(
	\theta';
	\mu(\cdot),
	\Sigma(\cdot)
	\right),
	\label{eq:pmmh:gaussianproposal}
\end{align}
where the mean $\mu$ and covariance $\Sigma$ are allowed to depend on $\theta_{k-1}$ and $u_{k-1}$.

The simplest Gaussian proposal is given by the random walk corresponding to $\mu=\theta_{k-1}$ and $\Sigma$ as some fixed covariance matrix.
However, it is known that this proposal scales unfortunately with the dimension of the state $\theta_k$.
For larger problem, it is better to make use of gradient and Hessian information when constructing the proposal.
This information can be used to introduce a mode-seeking behaviour by following the drift induced by the gradient as well as scale the proposal to the local curvature using the Hessian.

We can introduce this additional information into the proposal by considering a Langevin diffusion,
\begin{align}
	\dn \theta_t
	=
	- \frac{1}{2} \Sigma
	\nabla \log \pi(\theta) \big|_{\theta=\theta_t}
	+
	\sqrt{\Sigma}
	\dn B_t,
	\label{eq:pmmh:langevindiffusion}
\end{align}
where $B_t$ denotes a Brownian motion.
This choice is motivated by the diffusion having $\pi(\theta)$ as its stationary distribution under some mild assumptions.
Hence, samples from the target can be obtained by simply simulating the diffusion for a certain amount of time.
Moreover, we follow \cite{GirolamiCalderhead2011} and make use of the negative inverse Hessian of the log-target as $\Sigma$ to facilitate efficient sampling.

The first-order Euler discretisation of the Langevin diffusion \eqref{eq:pmmh:langevindiffusion} is given by the Gaussian proposal \eqref{eq:pmmh:gaussianproposal} with the statistics,
\begin{align}
	\mu(\theta_{k-1})
	=
	\theta_{k-1}
	+
	\frac{\epsilon^2}{2}
	H(\theta_{k-1})
	G(\theta_{k-1}),
	\qquad
	\Sigma(\theta_{k-1})
	=
	\epsilon^2
	H(\theta_{k-1}),
	\label{eq:pmmh:secondorder:proposal}
\end{align}
where $\epsilon > 0$ denotes a step size and introducing the notation
\begin{align*}
	G(\theta') = \nabla_{\theta} \log \pi(\theta) \big|_{\theta=\theta'}, \smallskip
	H^{-1}(\theta') = - \nabla^2_{\theta} \log \pi(\theta) \big|_{\theta=\theta'},
\end{align*}
for the gradient and negative inverse Hessian of the log-target, respectively.

Note that it is not possible to compute $G(\theta')$ and $H(\theta')$ in many situations when $\pi(\theta)$ cannot be evaluated.
However, these quantities can be estimated using e.g., SMC which make use of the auxiliary variables $u$ entering the proposal.
For example in SSMs, estimating $G(\theta)$ using $u$ by employing particle smoothing problem and the Fisher identity \citep[p.\ 352]{CappeMoulinesRyden2005}.

The Hessian can be approximated in a similar manner by the Louis identity \citep[p.\ 352]{CappeMoulinesRyden2005}.
The problem is that the accuracy of the Hessian estimates often is worse than for the gradient estimates.
Empirically, the estimates of the smallest diagonal and off-diagonal elements of the Hessian are often very noisy even with many particles in the SMC algorithm.
This introduce a computational bottleneck as each iteration of pmMH becomes prohibitively expensive.

\section{Quasi-Newton proposals}
\label{sec:hessian}
We propose to instead compute a local approximation of the Hessian using gradient information by leveraging qN methods.
This is useful as the gradient estimates are usually more accurate than the Hessian estimates and applying qN only amounts to a negligible computational overhead.
The use of qN can therefore serve as an alternative to the use of the Louis identity for SSMs as well as direct computation of the Hessian for some models when large amounts of data are available.

Hessian estimates obtained by qN methods \citep{NocedalWright2006} have a long history in the optimisation literature for finding the maximum or minimum of non-linear functions.
However, a major difference in this paper compared with most of the optimisation literature is that only noisy estimates of the gradients are available.
Most qN methods have not been developed for this situation, which can result in numerical instability in the Hessian estimates.
Such instability could result in that a candidate parameter far away from the posterior mode is accepted.

In this section, we address these issues by introducing a novel so-called trust-region method to encode the region of the target space in which the Hessian approximation is valid.
Furthermore, we investigate a quite recent alternative based on regularised LS, which can deal with the noisy gradient estimate in a natural manner.
Finally, we make use of limited-memory implementations leveraging the memory introduced in Section~\ref{sec:pmmh:memory}.

\subsection{SR1 with a trust-region}
\label{sec:hessian:sr1}
SR1 \citep[Chapter~6.2]{NocedalWright2006} is a common qN method which has been successfully applied to many different optimisation methods.
Compared with other qN methods, SR1 can give substantially better estimates of the Hessian and can enjoy more rapid convergence \citep{ConnGouldToint1991}.
Hence, it is a natural candidate for building proposals with pmMH but has previously not been considered in this context in the statistics or machine learning literature.

The SR1 algorithm is implemented as an iterative procedure which updates the current estimate of the inverse Hessian $H_l$ by
\begin{align}
	H_{l+1}
	&=
	H_l
	+
	\frac
	{(s_l - H_l g_l)(s_l - H_l g_l)^{\top}}
	{(s_l - H_l g_l)^{\top} g_l},
	\label{eq:sr1:update}
	\\
	s_l
	&\triangleq
	\theta_l - \theta_{l-1},
	\quad
	g_l
	\triangleq
	G(\theta_l)
	-
	G(\theta_{l-1}).
	\nonumber
\end{align}
where $s_l$ and $g_l$ denote the differences in the parameters and gradients between two iterations, respectively.
The name SR1 comes from that the update is a so-called rank-one update as the estimate is updated by a term which is an outer product of two vectors.

A limited-memory version of SR1 is obtained by iterating \eqref{eq:sr1:update} over the $M-1$ values of $s_l$ and $g_l$ corresponding to the memory \eqref{eq:pmmh:memory:def}.
The initial value of the Hessian estimate $H_0$ is computed as a scaled version of the identity matrix $\epsilon_0 \|G(\theta')\|_2^{-1} \mathbf{I}_{n_{\theta}}$, where $\epsilon_0$>0 is selected by the user and $G(\theta')$ denotes the gradient in the candidate parameter.

As discussed above, an SR1 method can become numerically unstable due to e.g., the noise entering the gradient estimates, rounding errors and loss of rank.
Therefore, an SR1 method is usually accompanied by a trust-region to encode the part of the target space in which the Laplace approximation implied by the proposal \eqref{eq:pmmh:secondorder:proposal} is believed to be accurate.

To incorporate this into pmMH, we propose to augment the proposal to be a product of the original proposal and a trust-region distribution,
\begin{align}
	\bar{q}(\theta' | \theta_{k-1})
	=
	q(\theta' | \theta_{k-1})
	\, \,
	q_{\text{trust}}(\theta' | \theta_{k-1}),
	\label{eq:qn:sr1trustregion}
\end{align}
where $q(\theta' | \theta_{k-1})$ is given by \eqref{eq:pmmh:secondorder:proposal}.
A natural choice for the trust-region distribution is given by
\begin{align*}
	q_{\text{trust}}(\theta' | \theta_{k-1})
	=
	\mathcal{N}(\theta'; \theta_{k-1}, \Lambda_k),
\end{align*}
where $\Lambda_k$ is some covariance matrix at iteration $k$.
In this paper, we make use of a similar approach to the hybrid method \citep{DahlinLindstenSchon2015a} and set $\Lambda_k$ to the sample covariance matrix of the Markov chain computed using a part of the burn-in iterations.

The choice of the trust-region distribution to be a Gaussian is natural since the product of two Gaussians is itself a Gaussian distribution.
Hence for this choice, we can rewrite \eqref{eq:qn:sr1trustregion} as a new Gaussian distribution with updated statistics, which makes it easy to sample from the proposal and to evaluate it point-wise.

Furthermore, this mechanism mimics the behaviour of a trust-region method in standard qN methods, which adapts its size depending on the length of the previous steps.
A simple analysis of \eqref{eq:qn:sr1trustregion} reveals that $\Lambda_k$ will determine the maximum scale of the region in which the local model is trusted.

\subsection{Least-squares}
\label{sec:hessian:ls}
Another method for estimating the Hessian is to make use of LS as proposed by \cite{WillsSchon2018} and \cite{HaeltermanDegrooteVanHueleVierendeels2009}
This can be accomplished by leveraging the so-called \textit{secant condition},
\begin{align}
	G(\theta_l)
	-
	G(\theta_{l-1})
	=
	H^{-1}
	(\theta_l - \theta_{l-1}),
\end{align}
for some pair of parameters $\theta_l$ and $\theta_{l-1}$.
Thus if $M-1$ such pairs are collected, the inverse Hessian of the log-target can be estimated by LS using
\begin{align}
	H_k
	Y_k
	&=
	S_k,
	\label{eq:ls:update}
	\\
	Y_k
	&\triangleq
	\Delta G_{k-1:k-M-1},
	\nonumber
	\\
	S_k
	&\triangleq
	\Delta \theta_{k-1:k-M-1},
	\nonumber
	\\
	\Delta v_{i:j}
	&\triangleq
	(v_i - v_{i+1}, \ldots, v_{j-1} - v_j).
	\nonumber
\end{align}
This LS method has some interesting advantages over traditional qN methods such as: (i) it handles the noise in the gradient estimates in a natural manner, (ii) it requires no initialisation of the Hessian estimate and (iii) it is straightforward to implement using existing commands in most computing environments.

Some care needs to be taken when applying \eqref{eq:ls:update} with estimates of the gradients computed using $u$.
That is, the noise in the estimates will influence future values of $\theta$ and therefore introduce problems with multi-colinearity.
This can be mitigated by using two different sets of auxiliary variables $u$, one set to compute the gradient estimate entering \eqref{eq:pmmh:secondorder:proposal} and another set to compute the estimate used in \eqref{eq:ls:update}.
Note that this adds no extra computational cost as these two estimates can be obtained in parallel.
We make this idea explicit later when discussing the implementation of qN proposals within pmMH in Algorithm~\ref{alg:pmmh:qn}.

Finally, it is possible to add regularisation of the Hessian estimator and to make use of updates using matrix factorisations to improve computational efficiency.
This is especially important when $n_{\theta}$ is large \citep{WillsSchon2018} or when the gradients are noisy to obtain stable estimates of the Hessian.
A regularised least squares estimate can be computed by
\begin{align}
	H_k
	=
	\left( \lambda I + Y_k Y_k^{\top} \right)^{-1}
	\left( \lambda \Lambda_k + Y_k S_k^{\top} \right),
	\label{eq:ls:update:reg}
\end{align}
where $\lambda > 0$ denotes the coefficient determining the strength of the regularisation.
Here, $\Lambda_k$ denotes the regularisation matrix at iteration $k$ which is computed in the same manner as the trust-region covariance for SR1.

In preliminary studies, we have seen that the use of regularisation improves the performance significantly by improving the numerical stability of the Hessian estimates.
Moreover, selecting the regularisation parameter $\lambda$ as $0.1$ results in good performance for many models.

\subsection{Implementation}
\label{sec:hessian:impdetails}
In this section, we describe how to implement qN as a proposal using pmMH with memory as described in Algorithm~\ref{alg:pmmh:memory}.
We make use of a sliding window of previous states of the Markov chain denoted by $\psi_{k,M} \triangleq \{\theta_i, G(\theta_i)\}_{i=k-M}^k$, which is the information about the parameters and gradients from the $M-1$ previous steps.
This is equivalent to sequentially applying the Markov kernel \eqref{eq:mh:within:gibbs:update} but is a more convenient formulation for implementation.

The qN proposal can then be expressed as
\begin{align}
	q \big( \theta' | \psi_{k,M} \big)
	&=
	\mathcal{N}
	\Big(
	\theta';
	\mu_{\text{qN}} \big( \psi_{k,M} \big),
	\epsilon^2 \Sigma_{\text{qN}} \big( \psi_{k, M} \big)
	\Big),
	\label{eq:pmmh:qn:proposal}
	\\
	\mu_{\text{qN}}(\psi_{k,M})
	&=
	\theta_{k-M}
	+
	\frac{\epsilon^2}{2}
	\Sigma_{\text{qN}}(\psi_{k,M})
	G(\theta_{k-M}),
	\nonumber
\end{align}%
\noindent where $\Sigma_{\text{qN}}(\psi_{k,M})$ denotes the qN approximation of the negative inverse Hessian of the log-target.
The mean of the proposal is centered at the parameter $M-1$ steps back, which is motivated by good performance in preliminary studies.
However, we are free to select any other parameter in the memory using a fixed or random scheme.
In~\cite{ZhangSutton2011}, the authors make use of the previous parameter, i.e., $\theta_{k-1}$.

\begin{algorithm}[!t]
    \caption{\textsf{Quasi-Newton (qN) proposal}}
    \footnotesize
	\textsc{Inputs:} $\psi_{k,M} \triangleq \{\theta_i, G(\theta_i)\}_{i=k-M}^k$ and $\delta > 0$.
	\\
	\textsc{Output:} $\theta'$.
	\algrule[.4pt]
	\begin{algorithmic}[1]
		\STATE Extract the unique elements from $\psi_{k,M}$ and sort them in ascending order (with respect to the log-target) to obtain $\bar{\psi}_{k,M}$.
		\STATE Set $\bar{M} = \textsf{size}(\bar{\psi}_{k,M}).$
		\IF{ $\bar{M} \geq 2$ }
			\STATE Initialise the Hessian estimate $H_0$.
			\FOR{$l=1$ to $\bar{M}$}
				\STATE Calculate $s_l$ and $u_l$ based on the $l$th pair in $\bar{\psi}_{k,M}$.
				\STATE [\textsf{SR1 update}] Compute \eqref{eq:sr1:update} to obtain $H_l$.
			\ENDFOR
			\STATE [\textsf{LS update}] Compute $H_{\bar{M}}$ by \eqref{eq:ls:update:reg}.
			\STATE Set $\Sigma_{\text{qN}}(\psi_{k,M}) = -H_{\bar{M}}(\theta')$.
			\ELSE
				\STATE Set $\Sigma_{\text{qN}}(\psi_{k,M}) = \delta \mathbf{I}_{n_{\theta}}$.
			\ENDIF
		\STATE Correct $\Sigma_{\text{qN}}(\psi_{k,M})$ to be PD using \eqref{eq:qn:spectralfix} if required.
		\STATE Sample from \eqref{eq:pmmh:qn:proposal} to obtain $\theta'$.
	\end{algorithmic}
	\label{alg:qnproposal}
\end{algorithm}

The procedure to sample from \eqref{eq:pmmh:qn:proposal} is given in Algorithm~\ref{alg:qnproposal}.
The pmMH method based on using a window of previous states is given in Algorithm~\ref{alg:pmmh:qn}.

The Hessian estimate $\Sigma_{\text{qN}}(\psi_{k,M})$ obtained using SR1 or LS can be negative semi-definite and therefore needs to be corrected.
In this paper, this is done by shifting the negative eigenvalues to be positive by a spectral method \citep[Chapter~3.4]{NocedalWright2006}.
This correction is computed by
\begin{align}
	\Sigma_{\text{qN}}(\psi_{k,M})
	=
	Q
	\tilde{\Lambda}
	Q^{-1},
	\label{eq:qn:spectralfix}
\end{align}
where $Q$ and $\tilde{\Lambda}$ denotes the matrix of eigenvectors and the diagonal matrix of the corrected eigenvalues of the Hessian $\Sigma_{\text{qN}}(\psi_{k,M})$, respectively.
The latter is computed as $\tilde{\lambda}_i = \max(\lambda_{\min}, |\lambda_i|)$, where $\lambda_i$ denotes the $i$th eigenvalue of $\Sigma_{\text{qN}}(\psi_{k,M})$.
Here, we introduce $\lambda_{\min}$ as some minimum size of the eigenvalues to ensure a non-singular matrix.
There are plenty of other methods for correcting the Hessian, e.g., the hybrid method \citep{DahlinLindstenSchon2015a} or modified Cholesky factorisations \citep[Chapter~3.4]{NocedalWright2006}.

\begin{algorithm}[!t]
    \caption{\textsf{Correlated pmMH using qN proposals}}
    \footnotesize
	\textsc{Inputs:} $K>0$, $\theta_0$, $q(\theta', u' | \theta_{k-1}, u_{k-1})$, $M>1$ together with Algorithm~\ref{alg:qnproposal} and its inputs.\\
	\textsc{Output:} $\{\theta_k\}_{k=1}^K$.
	\algrule[.4pt]
	\begin{algorithmic}[1]
		\STATE Sample $u_0 \sim m(u_0)$ and $\tilde{u}_0 \sim m(u_0)$,
		\STATE Compute $\widehat{\pi}(\theta_0 | u_0)$, $\widehat{G}_0 = \widehat{G}(\theta_0 | u_0)$ and $\tilde{\widehat{G}}_0 = \tilde{\widehat{G}}(\theta_0 | \tilde{u}_0)$.
		\FOR{$k=1$ to $K$}
			\STATE Sample the candidate auxiliary variables $u'$ and $\tilde{u}'$ given $u_{k-1}$ and $\tilde{u}_{k-1}$ by \eqref{eq:pmmh:cn:proposal}.
			\IF{$k < M$}
				\STATE Sample $\theta'$ \eqref{eq:pmmh:gaussianproposal} with $\mu(\theta_{k-1}) = \theta_{k-1}$ and $\Sigma(\theta_{k-1}) = \delta \mathbf{I}_{n_{\theta}}$.
			\ELSE
				\STATE Sample $\theta'$ using Algorithm~\ref{alg:qnproposal} and the memory given by $\psi_{k,M} \triangleq \{\theta_i, \tilde{\widehat{G}}_i\}_{i=k-M}^k$.
			\ENDIF
			\STATE Compute $\widehat{\pi}(\theta' | u')$, $\widehat{G}(\theta')$ and $\tilde{\widehat{G}}(\theta')$.
			\STATE Sample $\omega_k$ uniformly over $[0,1]$.
			\IF{$\omega_k \leq \min\{1,\alpha_k\}$ given by \eqref{eq:pmmh:aprob:simplified}.}
				\STATE Accept the candidate by assigning \\
				$\{
					\theta_k, u_k, \widehat{G}_k, \tilde{\widehat{G}}_k
				\}
					\leftarrow
					\{
					\theta', u', \widehat{G}(\theta'), \tilde{\widehat{G}}(\theta')
				\}$.
			\ELSE
				\STATE Reject the candidate by assigning \\
				$\{
					\theta_k, u_k, \widehat{G}_k, \tilde{\widehat{G}}_k
				\}
					\leftarrow
				\{
					\theta_{k-1}, u_{k-1}, \widehat{G}_{k-1}, \tilde{\widehat{G}}_{k-1}
				\}$.
			\ENDIF
		\ENDFOR
	\end{algorithmic}
	\label{alg:pmmh:qn}
\end{algorithm}

\section{Validity of the algorithm}
\label{sec:theory}
\label{sec:theory:validity}
We discussed the pmMH with memory framework in Section~\ref{sec:pmmh:memory} and gave a concrete example in Algorithm~\ref{alg:pmmh:qn} using the qN proposal.
In this section, we discuss the validity of Algorithms~\ref{alg:pmmh:memory} and \ref{alg:pmmh:qn}.

The validity depends on two different issues: (i) the conditions for when the Metropolis-within-Gibbs construction generates samples from $\bar{\pi}^M(\vartheta)$ and (ii) when the Markov kernel R generates samples from the sub-target $\bar{\pi}_i(\vartheta_i)$.
These issues are naturally connected as the requirements for (i) depends on (ii).
Hence, we start by considering the conditions for when $R$ leaves $\bar{\pi}_i(\vartheta_i)$ invariant and then address the validity of Metropolis-within-Gibbs.

The Markov kernel $R$ corresponding to pmMH targeting $\bar{\pi}_i(\vartheta_i)$ can be expressed as
\begin{align}
	R_i(\vartheta'_i, \vartheta_i)
	& =
	\alpha(\vartheta'_i, \vartheta_i)
	q(\vartheta_i' | \vartheta_i, \vartheta_{\setminus i})
	+
	\rho(\vartheta'_i, \vartheta_i)
	\delta_{\vartheta_i}(\vartheta'_i),
	\label{eq:pmmh:memory:kernel}
	\\
	\rho(\vartheta'_i, \vartheta_i)
	& =
	1
	-
	\dint
	\alpha(\xi, \vartheta_i)
	q(\xi | \vartheta_i, \vartheta_{\setminus i})
	\dd
	\xi,
	\nonumber
\end{align}
where $\alpha(\vartheta'_i, \vartheta_i)$ is given by \eqref{eq:theory:aprob} and $\rho(\vartheta'_i, \vartheta_i)$ denotes the rejection probability.
This Markov kernel generates a reversible Markov chain by construction if the detailed balance is fulfilled by the proposal and prior for $u$, i.e.\
\begin{align}
	q(\bar{u}'_i | \bar{u}_i) m_{\bar{\theta}_i}(\bar{u}_i)
	=
	q(\bar{u}_i | \bar{u}'_i) m_{\bar{\theta}'_i}(\bar{u}'_i)
	\label{eq:pmmh:memory:detailedbalance:aux}
\end{align}
which e.g., holds for the proposal in \eqref{eq:pmmh:cn:proposal}.

\begin{lemma}
	Under Assumption~1 in \cite{AndrieuRoberts2009} and given that the (ideal) MH algorithm targeting $\pi_i(\bar{\theta}_i)$ defines a $\phi$-irreducible and reversible Markov chain.
	Then, the Markov chain defined by $R_i(\vartheta'_i, \vartheta_i)$ in \eqref{eq:pmmh:memory:kernel} is also irreducible and reversible.
	Furthermore, we have in the total variational (TV) norm
	 \begin{align*}
		\|
			R^{k}_i(\cdot, \vartheta_{0, i})
			-
			\bar{\pi}_i(\vartheta_i)
		\|_{\text{TV}}
		\longrightarrow
		0,
	 \end{align*}
	 when $k \rightarrow \infty$ and for any $\vartheta_{0, i}$ such that $\bar{\pi}_i(\vartheta_{0, i}) > 0$.
	 \label{lemma:pmmh}
\end{lemma}
\begin{proof}
	Follows from Theorem~1 in \cite{AndrieuRoberts2009} together with the additional detailed balance condition for the proposal of $u$ in \eqref{eq:pmmh:memory:detailedbalance:aux}.
\end{proof}

The idea is to make use of the Metropolis-within-Gibbs construction in \eqref{eq:mh:within:gibbs:update}, which either systematically updates one component after another or randomly selecting the component to update.
The conditions for when such a scheme is valid has been investigated by numerous authors but the results in \cite{RobertsRosenthal2006} are particular of interest for this paper.
Especially, Corollary~19 gives the sufficient conditions for the Markov chain targeting the extended product target and the sub-chains being Harris recurrent which together with $\phi$-irreducibility gives an ergodic chain.

\begin{lemma}
	Given that the target $\bar{\pi}^M(\vartheta)$ is integrable over any combination of $\{\Theta_1, \mathcal{U}_1, \ldots, \Theta_M, \mathcal{U}_M\}$.
	Then, the Markov kernel $\bar{R}^{k}(\vartheta', \vartheta)$ given by the combination of \eqref{eq:mh:within:gibbs:update} and \eqref{eq:pmmh:memory:kernel} defines a Harris recurrent Markov chain.
	Hence, if the Markov chain also is $\phi$-irreducible, we have
	\begin{align}
		\|
			\bar{R}^{k}(\cdot, \vartheta_0)
			-
			\bar{\pi}^M(\vartheta)
		\|_{\text{TV}}
		\longrightarrow
		0,
		\label{eq:mwg:convergence}
	 \end{align}
	 when $k \rightarrow \infty$ and for any $\vartheta_{0}$ such that $\bar{\pi}^M(\vartheta_0) > 0$.
	 Furthermore, the same holds for all the sub-chains.
	 \label{lemma:mwg}
\end{lemma}
\begin{proof}
	Follows directly from Theorem~6 and Corollary~18 in \cite{RobertsRosenthal2006}.
\end{proof}

To summarise, we combine Lemmas~\ref{lemma:pmmh} and \ref{lemma:mwg} to show that the Markov chains generated using Algorithms~\ref{alg:pmmh:memory} and \ref{alg:pmmh:qn} converges to the sought target distribution.

\begin{proposition}
	Given the conditions of Lemmas~\ref{lemma:pmmh} and \ref{lemma:mwg}.
	The Markov chain defined by \eqref{eq:pmmh:memory:kernel} and it sub-chains are ergodic and have $\bar{\pi}^M(\vartheta)$ and $\bar{\pi}_i(\vartheta)$ as their stationary distributions, respectively.
\end{proposition}

As a consequence, the samples generated using this procedure $\{\vartheta_{k,i}\}$ are distributed according to $\bar{\pi}^M(\vartheta)$.
If we have that $\bar{\pi}_i(\vartheta) = \bar{\pi}(\vartheta)$, i.e., all the sub-chains target the same distribution, then the random samples are distributed according to $\bar{\pi}(\theta, u)$ which is the original pmMH target in \eqref{eq:extendedtarget}.

These random samples can be used to estimate expectations with respect to the extended target and the parameter posterior.
Focusing on the latter, we let $\varphi: \Theta \rightarrow \mathbb{R}$ denote a well-behaved integrable \textit{test function} with respect to the parameters $\theta$.
The expectation of $\varphi$ with respect to $\pi(\theta)$ can be obtained by marginalisation of $\bar{\pi}(\theta,u)$ resulting in
\begin{align}
	\bar{\pi}[\varphi]
	\triangleq
	\mathbb{E}_{\bar{\pi}}
	[\varphi(\theta)]
	=
	\dint
	\varphi(\theta)
	\pi(\theta)
	\dn \theta
	\label{eq:pmh:testfunction}
\end{align}
which unfortunately is intractable for most interesting problems.
However, it can be approximated by
\begin{align}
	\widehat{\bar{\pi}}^{MK}[\varphi]
	=
	\sum_{i=1}^M
	\sum_{k=1}^K
	\varphi(\bar{\theta}_{i,k}),
	\label{eq:mh:testfunction:approx:ensemble}
\end{align}
which by the ergodic theorem is a consistent estimator,
\begin{align*}
	\widehat{\bar{\pi}}^{MK}[\varphi]
	\stackrel{\text{a.s.}}{\longrightarrow}
	\bar{\pi}[\varphi],
	\qquad
	K \rightarrow \infty,
\end{align*}
for any finite memory length $M$.
Note that, we could thin the chain and retain only the samples $\theta_{i,k}$ for some $i \in \{1, \ldots, M\}$.
However, we opt for retaining all samples (even if they are correlated) by the standard argument against thinning \citep{Geyer1992,MacEachernBerliner1994}.
This is also supported by comparisons made using pilot runs.


\section{Numerical illustrations}
\label{sec:results}
In this section, we investigate the performance and properties of Algorithm~\ref{alg:pmmh:qn} using SR1 and LS methods when carrying out Bayesian parameter inference in three different models.
Furthermore, these two novel proposals are compared with four alternatives; pmMH0/1/2 \citep{DahlinLindstenSchon2015a,NemethSherlockFearnhead2016} and Algorithm~\ref{alg:pmmh:qn} using damped BFGS \citep{DahlinWillsNinness2018b}.

The pmMH0/1 method corresponds to using \eqref{eq:pmmh:gaussianproposal} without/with gradient information and replacing the Hessian with a constant matrix $\Sigma$, which is usually an estimate of the posterior covariance obtained by pilot runs.
The pmMH2 method is essentially the same as pmMH but with the Hessian computed without the use of memory by direct computations, importance sampling or SMC.

We compare the mixing of the Markov chains generated using the different MCMC methods using the inefficiency factor (IF) also known as the \textit{integrated auto-correlation time}.
The IF is computed by
\begin{align}
	\mathsf{IF} = 1 + 2\sum_{k=2}^{\infty}\mbox{corr}\{\theta_1,\theta_k\}.
	\label{eq:pmh:iact}
\end{align}
where a value of one corresponds to perfect (independent) sampling from the target.
However, the IF computation cannot be carried out as only finite realisations of the Markov chains are available.

Instead, the IF is approximated by the sum of the first $250$ empirical autocorrelation function (ACF) coefficients.
The time per effective samples (TES) is also computed by multiplying the IF with the time required per pmMH iteration.
The TES value can be interpreted as the time required to obtain one uncorrelated sample from the posterior, which takes into account the varying computational cost of the benchmarked methods.

All the implementation details are summarised in \ref{app:impdetails} and the source code and supplementary material are available via GitHub \url{https://www.github.com/compops/pmmh-qn} and Docker Hub (see \url{README.md} in GitHub repository).

\begin{sidewaysfigure}[p]
	\centering
	\includegraphics[width=\textwidth]{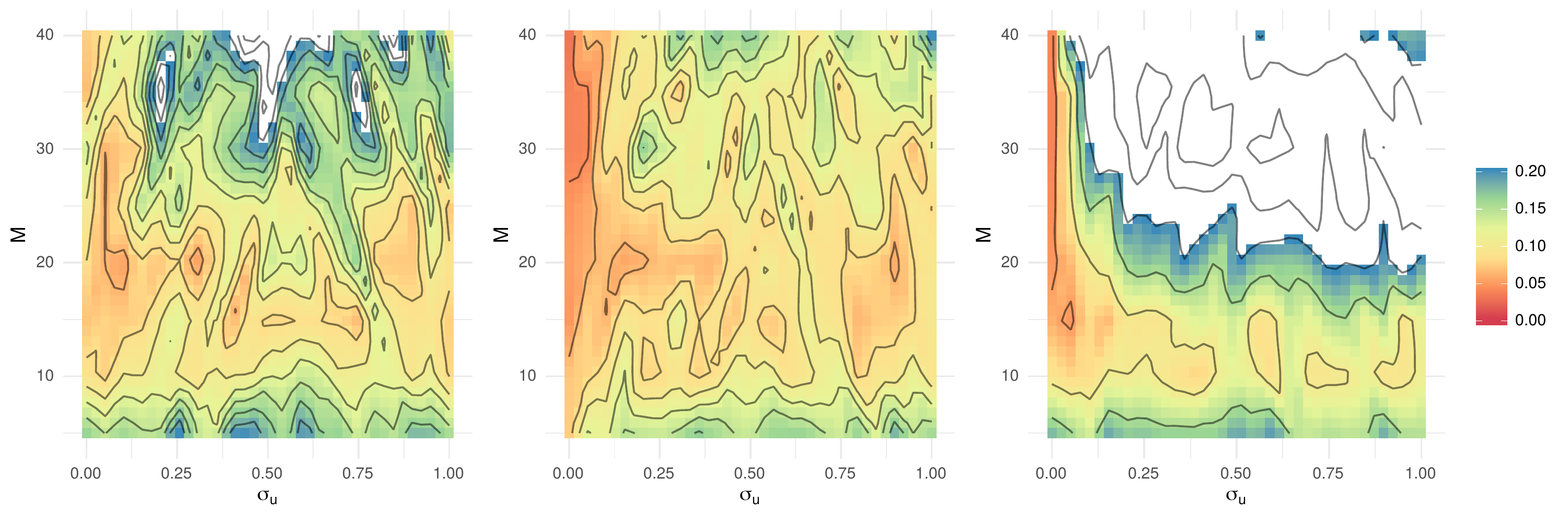}
	\caption{\footnotesize The TES (in seconds) with respect to $\sigma_u$ and $M$ for BFGS (left), LS (middle) and SR1 (right) computed using $5$ Monte Carlo runs on the random effects model with synthetic data. The white regions indicate areas with a TES larger than 0.2 seconds.}
	\label{fig:example1-qn}
\end{sidewaysfigure}

\subsection{Selecting correlation and memory length}
\label{sec:results:random}
We begin by investigating the impact of the two main user choices $M$ and $\sigma_u$ on the performance of Algorithm~\ref{alg:pmmh:qn} using BFGS, LS and SR1 methods.
The aim is to provide some guidelines for the user to tune the acceptance probability to obtain good a small TES, comparable to existing rule-of-thumbs \citep{NemethSherlockFearnhead2016}.

To this end, we consider estimating the parameter posterior \eqref{eq:parameter:posterior} for the random effects model given by
\begin{align*}
	x_t
	\sim
	\mathcal{N}(x_t | \mu, \sigma^2
	),
	\qquad
	y_t | x_t
	\sim
	\mathcal{N}(y_t | x_t, 1),
\end{align*}
with the parameters $\theta = \{ \mu, \sigma \}$ given the observations $y_{1:T} = \{y_t\}_{t=1}^T$.
The parameter posterior is intractable for this model as the likelihood depends on the latent states $x_{1:T}$.
However, it is possible to obtain point-wise estimates of the log-likelihood and its gradients using importance sampling, see \ref{app:impdetails:random} for details.

We generate a data set with $T=100$ observations using $\theta=\{1.0, 0.2\}$ and run an importance sampler with correlated samples together with Algorithm~\ref{alg:pmmh:qn} to estimate the posterior.

Figure~\ref{fig:example1-qn} summaries the median TES for the three different qN methods.
First, we note that SR1 and LS seem to provide better maximum performance compared with BFGS.
Second, LS seems to be somewhat more robust to the user choices than BFGS and SR1.
Third, the optimal value of $M$ seems to be around $20-30$ depending on the method and a small non-zero value of $\sigma_u$ is preferable.
Finally, smaller values of $M$ are preferable when $\sigma_u$ is large, which is natural as this would result in less accurate gradient estimates.

\begin{figure}[p]
	\centering
	\includegraphics[width=0.75\textwidth]{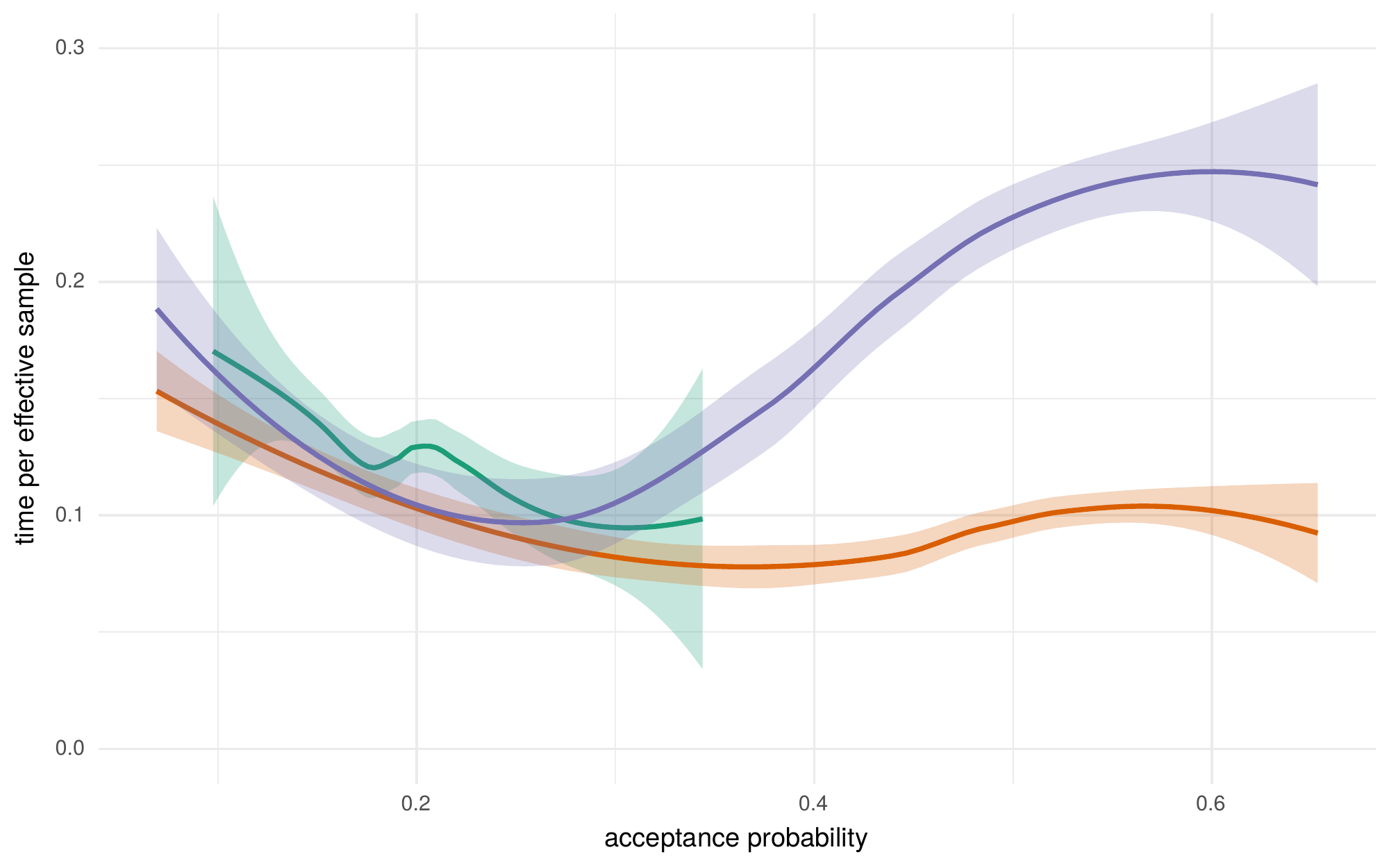}
	\caption{\footnotesize A local polynomial regression estimate of the TES as a function of the acceptance probability for BFGS (green), LS (orange) and SR1 (purple) for the random effects model with synthetic data.}
	\label{fig:example1-qn-aprob}
\end{figure}

Figure~\ref{fig:example1-qn-aprob} presents the TES as a function of the acceptance rate based on the data in Figure~\ref{fig:example1-qn}.
The three different methods behave quite differently over the interval.
For example, LS is less sensitive to the acceptance probability than the SR1 update.

Optimal acceptance probabilities are found around $0.2-0.3$ (BFGS), $0.2-0.5$ (LS) and $0.2-0.3$ (SR1) as they minimise the TES.
Overall, the LS update seems to be the preferable choice to obtain good performance with the simplest tuning.

\subsection{Logistic regression with sub-sampling}
\label{sec:results:higgs}
We continue by investigating the accuracy of the Hessian estimate obtained qN as well as benchmarking the resulting performance of pmMH in a high-dimensional setting with $22$ parameters.
To this end, we consider a problem from physics namely the detection of Higgs bosons as discussed in \cite{BaldiSadowskiWhiteson2014}.
A simple approach to construct such a classifier is to make use of a logistic regression model given by
\begin{align*}
	y_t | x_t
	\sim
	\mathcal{B}(y_t | p_t),
	\quad
	p_t
	=
	\left[
		1 +
		\exp
		\left(
			-\beta x_t^{\top}
		\right)
	\right]^{-1},
\end{align*}
where $\mathcal{B}(p)$ denotes a Bernoulli distributed random variable with success probability $p$.
Here, $x_t$ and $\beta$ denote $22$ covariates for observation $t$ and the regression coefficients (including an intercept), respectively.

Again the problem amounts to the computation of the posterior \eqref{eq:parameter:posterior} of the parameters $\beta$ given the data $\{y_t, x_t\}_{t=1}^T$.
For this model, we can compute the log-likelihood and its gradients in closed form, but
this is computationally prohibitive as the data set is large.

\begin{figure}[p]
	\centering
	\includegraphics[width=0.75\textwidth]{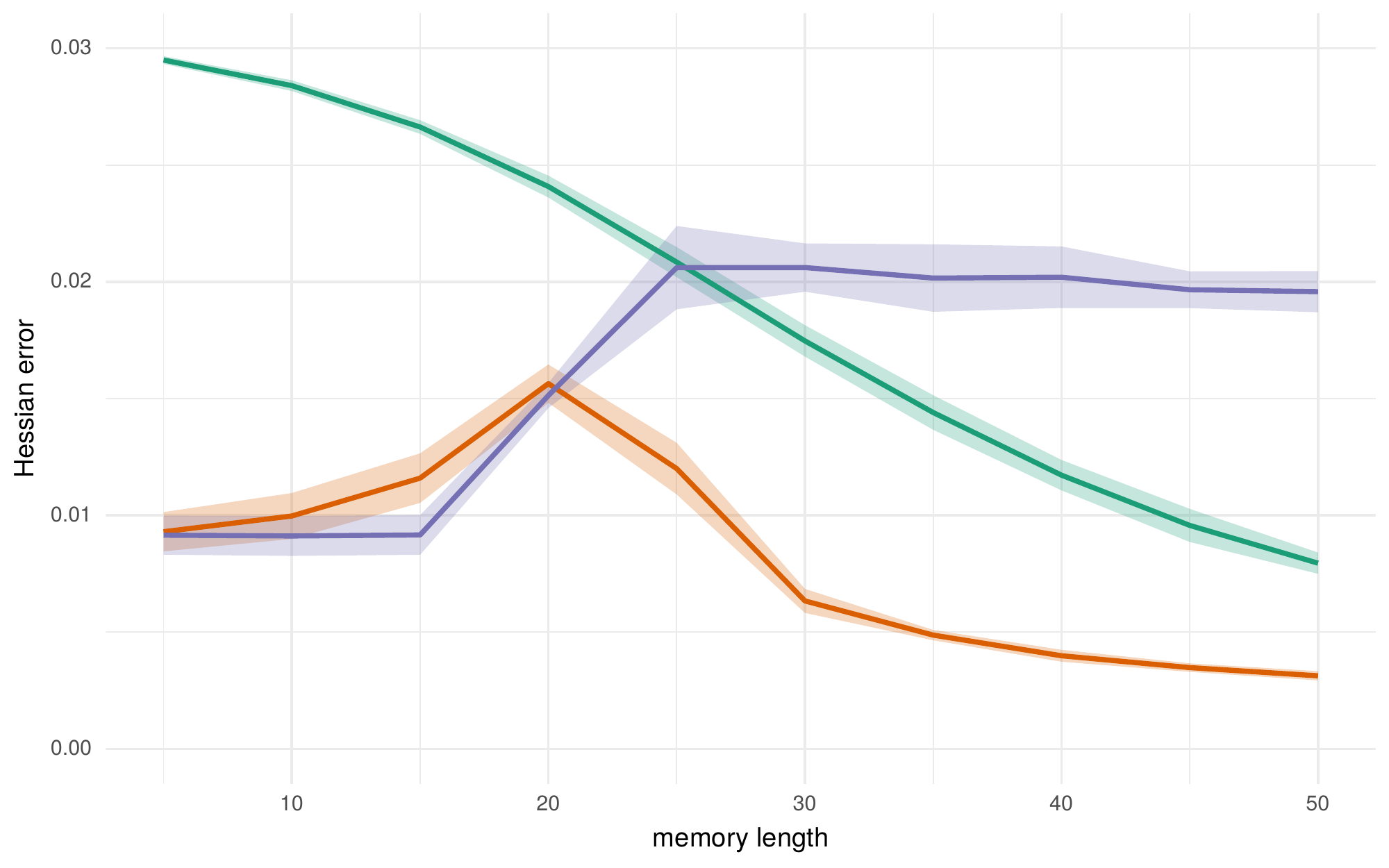}
	\caption{\footnotesize Median Frobenius norm error in the Hessian estimate from BFGS (green), LS (orange) and SR1 (purple) over $7,000$ iterations and 10 Monte Carlo runs on the logistic regression model using the Higgs data set. The shaded areas indicate the range between the $25\%$ and $75\%$ quantiles.}
	\label{fig:example2a-higgs}
\end{figure}

Instead, we make use of a sub-sampling approach, where the log-likelihood and its gradients are computed on a random subset of the data at each iteration.
This results in a proposal known as stochastic gradient descent (SGD) which previously has been applied to MH by \cite{WellingTeh2011}.
However, a naive method to sub-sampling results in a large variance in the estimates of the log-likelihood and large IF.

To mitigate this problem, we propose a possibly novel method which makes use of stratified resampling \cite{DoucetJohansen2011} together with correlated random numbers.
The Gaussian auxiliary variables $u$ are transformed into uniform variables using a CDF transform, i.e., by evaluating the Gaussian CDF in $u$.
These transformed auxiliary variables are then used within the stratified resampling algorithm to determine the subset.
More advanced methods are available in the literature and leveraged within Algorithm~\ref{alg:pmmh:qn}, see \cite{BardenetDoucetHolmes2017} and \cite{QuirozVillaniKohn2016}.

We begin by investigating the \emph{accuracy of the Hessian estimates} by comparing the estimates obtained by the BFGS, LS and SR1 with the true Hessian for $10$ independent Markov chains.
Each of the chains are run using a pmMH2 algorithm and the Hessian estimates using the three qN methods are computed at each iteration for a range of different memory lengths.
The error is computed as the mean of the Frobenius norm between the Hessian estimate and the true Hessian obtained by direct computation.

Figure~\ref{fig:example2a-higgs} presents the resulting median errors for the three different methods with varying memory lengths.
Note that the LS update attains the lowest error when $M$ increase past the number of parameters $p$, which is natural as $M > p$ is required to obtain a unique solution to the LS problem.
Moreover, the error in the Hessian estimate obtained by SR1 seems to be constant for large $M$, which could be due to the fact that this update is only valid for exact gradients.

We continue by examining the performance of complete implementations of pmMH using qN proposals for \emph{estimating the parameter posterior} of $\beta$.
Table~\ref{tbl:results:higgs} presents some performance statistics computed as the median over $10$ independent runs.
We have tried to match the acceptance probability with the recommendations above but this was difficult.

\begin{table}[t]
    \footnotesize
	\begin{center}
	\begin{tabular}{l|ccc|cc}
	\toprule
	& & & & \multicolumn{2}{c}{Time} \\
	\cmidrule(r){5-6}
	Alg. & Acc. & Cor. & mean \textsf{IF} &  Iter. & Samp. \\
	\midrule
	pmMH0  & 0.06 & -      & $372 \pm 27$  & 6 & 2.1 \\
	pmMH2  & 0.12 & 0.00   & $163 \pm 33$  & 29 & 4.7 \\
	\midrule
	pmMH-BFGS  & 0.19 & 0.00    & $33 \pm 19$           & 15 & 0.5 \\
	pmMH-LS    & 0.14 & 0.93    & $\mathbf{15} \pm 5$   & 13 & $\mathbf{0.2}$\\
	pmMH-SR1   & 0.26 & 1.00    & $\mathbf{15} \pm 4$   & 14 & $\mathbf{0.2}$ \\
	\bottomrule
	\end{tabular}
	\end{center}
	\caption{\footnotesize Performance statistics (acceptance rate, correction rate, IF and TES) as the median over $10$ Monte Carlo runs for different proposals in pmMH. The time per iteration is given in milliseconds and the TSE is given in seconds.}
	\label{tbl:results:higgs}
\end{table}

The smallest IF and TES are attained by LS and SR1, which corresponds well with results in the previous illustration.
In fact, all the IF values for the qN proposals are much smaller than for pmMH0/2.
This is the result of that the behaviour of the ACF is quite different for qN proposals compared with standard proposals due to the introduction of the memory.

\begin{figure}[p]
	\centering
	\includegraphics[width=\textwidth]{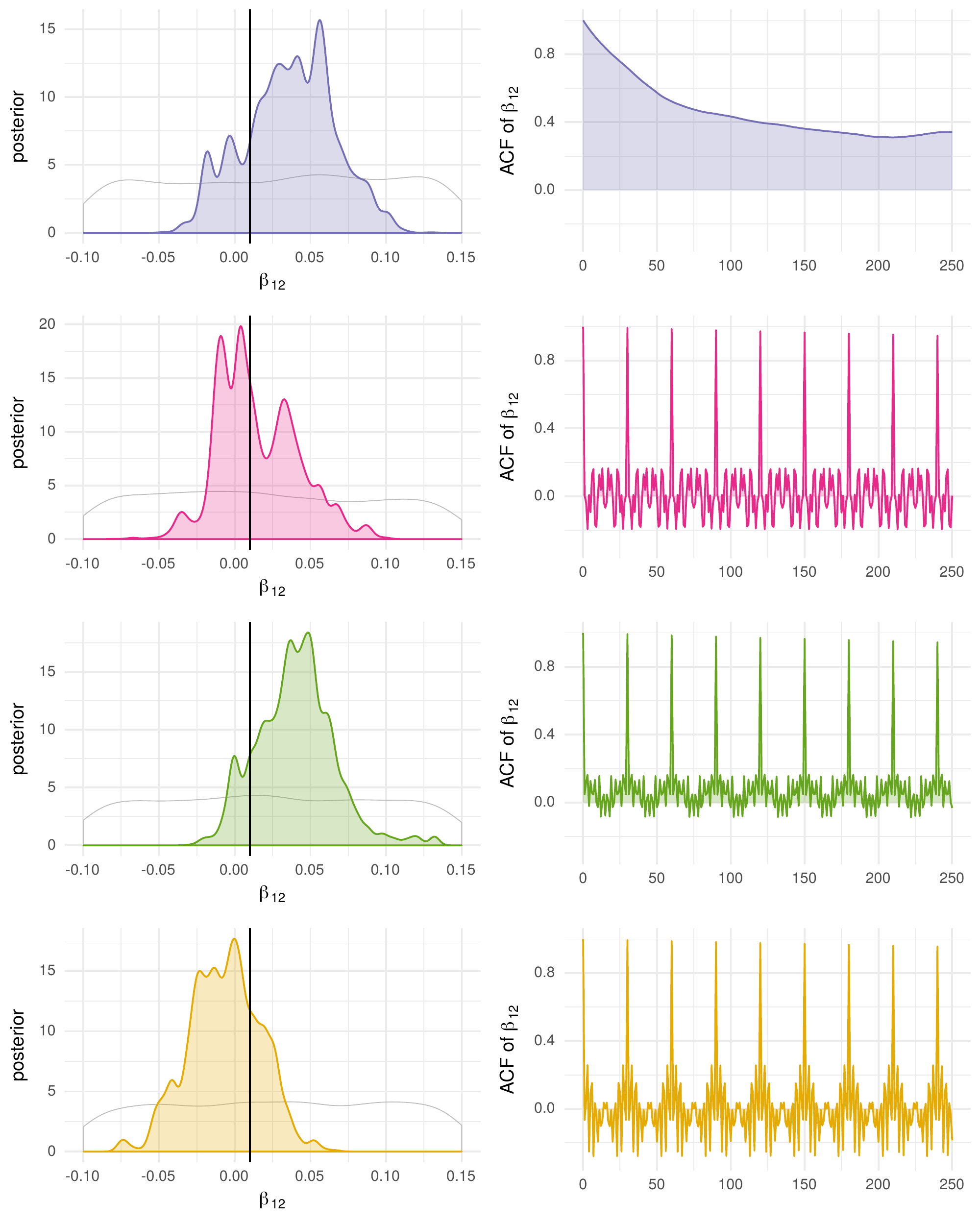}
	\caption{\footnotesize The posterior estimates (left) and empirical ACF (right) for $\beta_{12}$ obtained by pmMH2 (purple) and Algorithm~\ref{alg:pmmh:qn} using BFGS (magenta), LS (green) and SR1 (yellow) in the logistic regression model using the Higgs data set. The vertical and gray lines indicate estimates obtained by SGD and the prior distribution, respectively.}
	\label{fig:example2b-higgs}
\end{figure}

Figure~\ref{fig:example2b-higgs} presents the corresponding posterior estimates and empirical ACFs of the Markov chains for $\beta_{12}$.
The posteriors are all located around the parameter estimate obtained by SGD, which indicates that all the chains have converged.
The empirical ACFs for the qN proposals have a periodic behaviour as the proposal is centered around the state $M-1$ steps back.

An alternative to IF is to compare the variance of the posterior estimates, which is connected to the Monte Carlo error.
We compute the mean over the relative posterior variances with respect to pmMH2 and then take the median over the $10$ Monte Carlo runs to obtain: $1.37$ (BFGS), $1.08$ (LS) and $0.72$ (SR1) for the three qN proposals, respectively.

We conclude that all four proposals seem to give similar posterior variance.
However, the computational time for pmMH2 is about twice as large as for the qN proposals and also grows faster in terms of $T$.
This is the result of sums of outer products appear in the expressions for the Hessian.

Finally, recall that the Monte Carlo error typically scales as $1/\sqrt{K}$.
We therefore have a further decrease by $30\%$ of the relative variance for a fixed computational budget when taking the computational time into account.
Hence, the qN proposals are possibly a better choice in terms of performance compared with pmMH2 for this particular model.

\subsection{Stochastic volatility model for Bitcoin data}
\label{sec:results:sv}
We now investigate the performance of Algorithm~\ref{alg:pmmh:qn} on a model where it is quite easy to accurately estimate the gradient using SMC but difficult to estimate the Hessian in the same manner.
Here, the use of qN proposals could serve as a good alternative to pmMH2.

To illustrate this, we consider a stochastic volatility model with leverage to capture the behaviour of the Bitcoin log-returns presented in Figure~\ref{fig:example3-sv}.
This model can be expressed by
\begin{align*}
	x_0
	&\sim
	\mathcal{N}
	\left(
		\mu, \sigma_v^2 (1-\phi^2)^{-2}
	\right),
	\\
\begin{bmatrix}	x_{t+1} \\ y_{t} \end{bmatrix} \Bigg| x_t
	&\sim
	\mathcal{N}
	\left(
	\begin{bmatrix}	\mu + \phi( x_t - \mu ) \\ 0 \end{bmatrix},
	\begin{bmatrix}	\sigma_v^2 & \rho \\ \rho & \exp(x_t) \end{bmatrix}
	\right),
\end{align*}
where parameters of the model are $\theta=\{\mu, \phi, \sigma_v, \rho\}$.
The aim is to estimate the posteriors of $\theta$ and $x_{0:T}$ given observed log-returns $y_{1:T}$.

\begin{table}[t]
    \footnotesize
	\begin{center}
	\begin{tabular}{l|ccc|cc}
	\toprule
	 & & & & \multicolumn{2}{c}{Time} \\
	\cmidrule(r){5-6}
	 Alg. & Acc. & Cor. & max \textsf{IF} &  Iter. & Samp. \\
	\midrule
	pmMH2  & 0.64 & 1.00       & $429 \pm 22$       & 101 & 44.0 \\
	\midrule
	pmMH-BFGS  & 0.08 & 0.00    & $\mathbf{4} \pm 3$  & 18 & $\mathbf{0.1}$ \\
	pmMH-LS    & 0.17 & 0.40    & $7 \pm 2$           & 17 & $\mathbf{0.1}$ \\
	pmMH-SR1   & 0.21 & 0.72    & $6 \pm 3$           & 18 & $\mathbf{0.1}$ \\
	\bottomrule
	\end{tabular}
	\end{center}
	\caption{\footnotesize Performance statistics (acceptance rate, correction rate, IF and TES) as the median over $10$ Monte Carlo runs for different proposals in the stochastic volatility model using different proposals in pmMH. The time per iteration is given in milliseconds and the TSE is given in seconds.}
	\label{tbl:results:sv}
\end{table}

We make use of an SMC algorithm known as the fixed-lag particle smoother \citep{OlssonCappeDoucMoulines2008} to estimate the log-target and its gradient using the Fisher identity.
There are many other alternatives but this smoother is fast and reasonable accurate, see \cite{DahlinLindstenSchon2015a} and \cite{LindstenSchon2013}.
It is also possible to employ the Louis identity to compute an estimate of the Hessian as discussed in Section~\ref{sec:proposals:parameters} to use within the pmMH2 proposal and this is used as a comparison with qN proposals.

Table~\ref{tbl:results:sv} presents the median performance using the different proposals.
As before, a direct comparison between pmMH2 and the qN proposals is not possible but the three qN proposals achieve similar performance.

Figure~\ref{fig:example3-sv} presents the posterior estimates for $x_{0:T}$ and $\theta$ obtained by pmMH2 using the Louis identity and Algorithm~\ref{alg:pmmh:qn} using LS.
Here, we have matched the computational budget for the two methods and therefore the estimates obtained by pmMH2 use $2,673$ samples and the corresponding estimates from pmMH-LS makes use of $15,000$ samples.
This is the result of that the computations required by the Louis identity are slow even when implemented in \texttt{C} compared with the \texttt{Python} implementations of the computations required by the qN proposals.

We clearly see that pmMH2 struggles to explore the posterior for most parameters compared the other two proposals.
The two qN proposals based on BFGS and SR1 give essentially the same estimates and are therefore not presented in the figure.
The correlation parameter $\rho$ seems to be difficult to estimate and therefore the three posteriors differ.
Based on these results, we conclude that the qN proposals enjoy superior performance to pmMH2 and are at the same time simpler to implement than the Louis identity.

\begin{figure}[p]
	\centering
	\includegraphics[width=\textwidth]{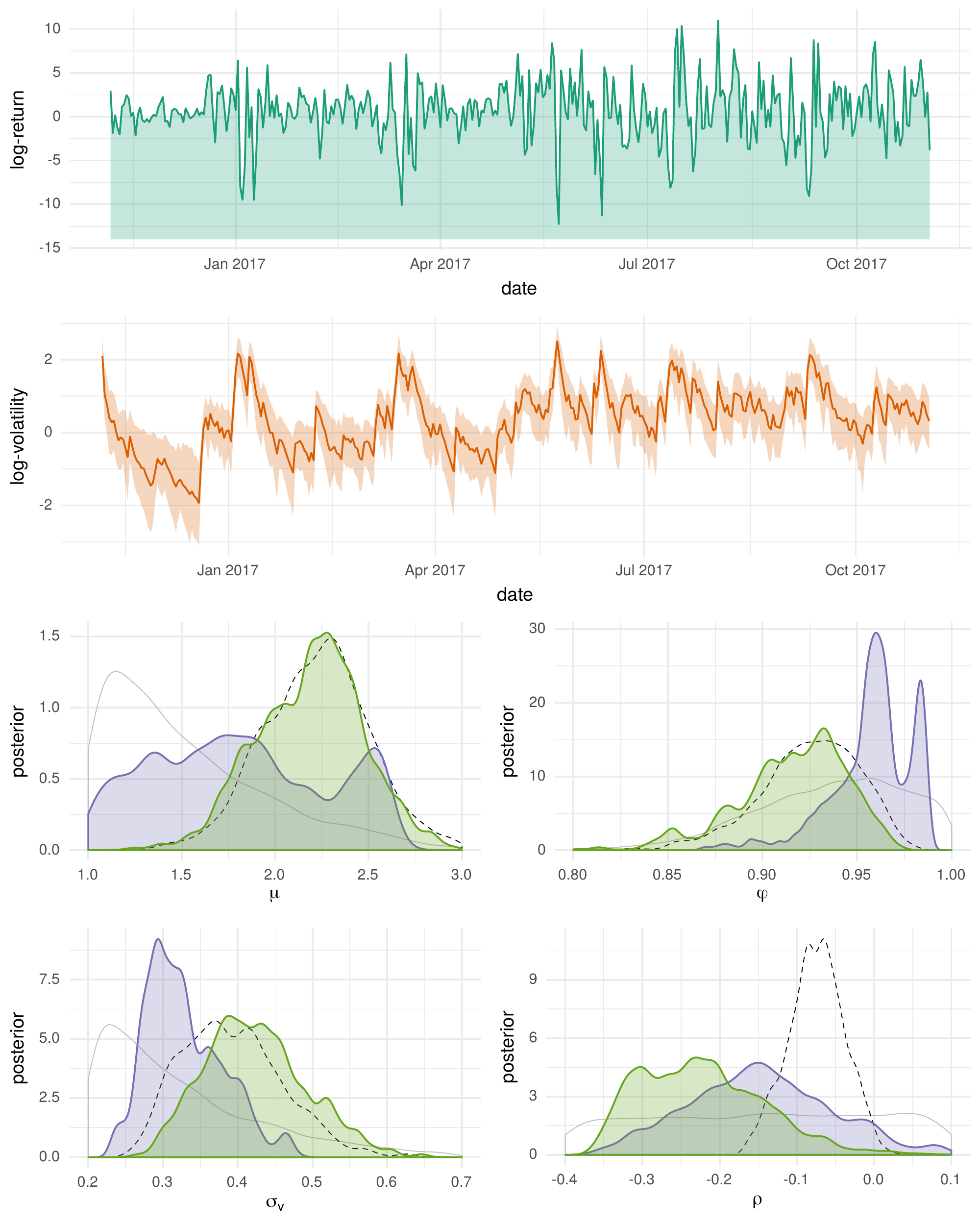}
	\caption{\footnotesize Upper: the log-returns of Bitcoin in terms of US Dollars during one year. Middle: the log-volatility estimate with $95\%$ high posterior density intervals obtained using Algorithm~\ref{alg:pmmh:qn} with LS. Lower two rows: the estimate of the posteriors for $\mu$ (Upper left), $\phi$ (upper right), $\sigma_v$ (lower left) and $\rho$ (lower right) using: pmMH0 (dashed), pmMH2 (purple) and Algorithm~\ref{alg:pmmh:qn} LS (green). The gray lines indicate the prior distributions.}
	\label{fig:example3-sv}
\end{figure}

\section{Conclusions}
\label{sec:conc}
The pmMH framework is a general methodology for exploring target distributions using Markov chains.
In this paper, we propose and prove the validity of an extension to pmMH allowing for using information from multiple previous iterations of the Markov chain when constructing the proposal.
We make use of this together with qN methods known as SR1 and regularised LS to obtain efficient proposals within pmMH.

The numerical illustrations indicate that SR1 and LS can obtain accurate estimates of the Hessian with a small computational overhead.
Hence, they can outperform direct computational of the Hessian in data-rich scenarios and when using SMC methods for estimating the Hessian.
The LS method seems to be the best choice with good overall performance and being quite robust to user choices.
Additionally, it enjoys well-understood statistical properties when the gradients are noisy, which current is not the case of BFGS and SR1.

These are interesting results as Hessian-based proposals can efficiently sample from high-dimensional targets \citep{GirolamiCalderhead2011}.
Hence, pmMH with qN proposal can be an alternative when pseudo-marginal HMC cannot be used \citep{LindstenDoucet2016}, which is the case for the models in Sections~\ref{sec:results:higgs}-\ref{sec:results:sv} or when the Hessian is difficult or time-consuming to estimate or compute directly.

Future work includes exploration of other types of proposals which makes use of the memory introduced in pmMH, see e.g., \cite{CaiMeyerPerron2008} and \cite{GilksRobertsGeorge1994}.
It is also possible to extend pseudo-marginal \citep{LindstenDoucet2016} with memory, which already has been proposed for HMC \citep{ZhangSutton2011,SimsekliBadeauCemgilRichard2016}.
Finally, more detailed analysis of the statistical properties of qN proposals is required to derive rule-of-thumbs to simplify calibration similar to \cite{DoucetPittKohn2015} .

\section*{Acknowledgements}
This work was supported by the Australian Research Council Discovery Project DP140104350. The authors would like to thank Prof.\ Thomas Sch\"{o}n and Dr.\ Fredrik Lindsten for valuable comments on an earlier draft.

\bibliographystyle{plainnat}
\bibliography{dahlin}

\appendix
\section{Implementation details}
\label{app:impdetails}
This section discusses the details of the implementations used to generate the illustrations in Section~\ref{sec:results}.
These details are also found in the source code and supplementary material available at \url{https://www.github.com/compops/pmmh-qn}.

\subsection{Selecting step sizes}
\label{app:impdetails:step}
We make use of an adaptive method to select the step size in the qN proposal \eqref{eq:pmmh:qn:proposal}.
We follow \cite{AndrieuThoms2008} and adapt $\epsilon_k$ such that the acceptance rate is around some target acceptance rate $\alpha^{\star}$ using the update
\begin{align}
\log \epsilon_k
=
\log \epsilon_{k-1}
+
\gamma_k
\left(
\alpha_{k-1}
-
\alpha^{\star}
\right),
\label{eq:adapt:step:size}
\end{align}
where $\gamma_k$ is selected so that the Robbins-Monro conditions are fulfilled, e.g., by $\gamma_k = k^{-\eta}$ for some $0 < \eta < 1$.
Note that this choice fulfils the requirement of diminishing adaption to obtain a valid MCMC scheme.

\subsection{Selecting correlation and memory length}
\label{app:impdetails:random}
We generate $T=100$ samples from the random effects model with parameters $\{\mu, \sigma\} = \{1.0, 0.2\}$.
The priors for the parameters in the model are
\begin{align*}
	\mu \sim \mathcal{N}(0, 1),
	\quad
	\sigma \sim \mathcal{C}_+(0, 1),
\end{align*}
where $\mathcal{C}_+(a, b)$ denotes the half-Cauchy distribution (over the positive real axis) with location $a \in \mathcal{R}$ and scale $b > 0$.
A reparametrization of the model is done to make all the parameters in the Markov chain unconstrained (able to assume any real value) given by
\begin{align*}
	\sigma_v = \exp(\tilde{\sigma}_v),
\end{align*}
where $\tilde{\theta}=\{\mu, \tilde{\sigma}_v\}$ are the new states of the Markov chain.
This change of variables introduces a Jacobian term into the acceptance probability,
\begin{align*}
	\frac{\sigma_v'}{\sigma_{v, k-1}}.
\end{align*}

We initialise the Markov chain in the true parameters for simplicity.
The Hessian $0.1 \cdot \mathbf{I}_2$ is used to scale the random walk for the initialisation phase (the first $M$ iterations) and as $\Lambda_k$ during the burn-in.
The pmMH algorithm is run for $K=30,000$ iterations with the first $3,000$ iterations discarded as burn-in.
The entire burn-in is used to estimate the posterior covariance, which is used as $\Lambda_k$ for the remaining iterations.

The step sizes are adapted using \eqref{eq:adapt:step:size} with $\eta=0.5$ and initial step size $0.1$ (BFGS), $0.15$ (LS) and $0.25$ (SR1).
These are calibrated using pilot runs to obtain reasonable mixing.
We set $\lambda=0.1$ for the regularisation of LS.

The log-target and its gradients and Hessian are computed using a correlated importance sampler with the prior for the latent process as the instrumental distribution and $N=100$ particles. See the supplementary materials or \cite{DeligiannidisDoucetPitt2017} for the details..

\subsection{Logistic regression with sub-sampling}
\label{app:impdetails:higgs}
We make use of the first $110,000$ observations and $21$ covariates from the data set downloaded from:
\url{https://archive.ics.uci.edu/ml/datasets/HIGGS}.
An intercept $\beta_0$ is included in the model by adding an appropriate column/row in the design matrix.
The prior for the parameters is $\beta_l \sim \mathcal{N}(0, 1)$ and all parameters are initialised at zero.
Each sub-sample consists of $5\%$ of the data set, i.e., $5,500$ observations.
The Hessian $0.01 \cdot \mathbf{I}_{22}$ is used to scale the random walk for the initialisation phase (the first $M$ iterations) and as $\Lambda_k$ during the burn-in.

The pmMH algorithm using the same settings as for the random effects model but with $M=40$ and $\sigma_u=0.05$.
The pmMH0 proposal makes use of step size $0.27$ and the covariance matrix estimated using Algorithm~\ref{alg:pmmh:qn} with LS during an earlier run is used as $\widehat{\Sigma}$.
The pmMH2 proposal makes use of step size $0.5$.
The step size for the qN proposals are adapted as before but with initial step size $0.1$ target and  acceptance rates $0.2$ (BFGS), $0.2$ (LS) and $0.3$ (SR1).
The SGD estimates are obtained using Lightning for Python \url{http://contrib.scikit-learn.org/lightning/index.html}.

\subsection{Stochastic volatility model for Bitcoin data}
\label{app:impdetails:sv}
The log-returns for Bitcoin are computed by $y_t = 100 [ \log(s_{t}) - \log(s_{t-1}) ]$, where $s_t$ denotes the daily exchange rates versus the US Dollar obtained from \url{https://www.quandl.com/BITSTAMP/USD}.
The priors for the parameters in the model are
\begin{align*}
	&\mu \sim \mathcal{N}(0,1^2), \quad
	\phi \sim \mathcal{TN}_{(-1,1)}(0.95,0.05^2), \\
	&\sigma_v \sim \mathcal{G}(2,10), \quad
	\rho \sim \mathcal{TN}_{(-1,1)}(0,1),
\end{align*}
where $\mathcal{TN}_{(a,b)}(\mu, \sigma^2)$ denotes a truncated Gaussian distribution on $[a,b]$ with mean $mu$ and standard deviation $\sigma$ and $\mathcal{G}(a,b)$ denotes the Gamma distribution with mean $a/b$.
Again, we reparameterise the model to make all the parameters in the Markov chain unconstrained (able to assume any real value) by
\begin{align*}
	\phi = \text{tanh}(\tilde{\phi}), \quad \sigma_v = \exp(\tilde{\sigma}_v) , \quad \rho = \exp(\tilde{\rho}),
\end{align*}
where $\bar{\theta}=\{\mu, \tilde{\phi}, \tilde{\sigma}_v, \tilde{\rho}\}$ are the new states of the Markov chain.
This change of variables introduces a Jacobian term into the acceptance probability,
\begin{align*}
	\frac{1 - (\phi')^2}{1 - \phi_{k-1}^2}
	\frac{\sigma_v'}{\sigma_{v, k-1}}
	\frac{1 - (\rho')^2}{1 - \rho_{k-1}^2}.
\end{align*}

We initialise the Markov chain in $\theta_0 = \{2.0, 0.9, 0.4, -0.2\}$.
The Hessian $\mathsf{diag}(0.01, 0.01, 0.01, 0.001)$ is used to scale the random walk for the initialisation phase (the first $M$ iterations) and as $\Lambda_k$ during the burn-in.
The pmMH algorithm is run using the same settings as before but with $\sigma_u=0.5$.
The pmMH2 proposal makes use of step size $0.8$.
The step sizes for the qN proposals are adapted as before but with initial step size.

The log-target and its gradients are estimated using a fixed-lag particle smoother using correlated random numbers with lag $\Delta=10$.
The algorithm is given in the supplementary materials and \cite{DahlinLindstenSchon2015a} and $N=75$ particles are used which is approximates $T^{2/3}$ as recommended by \cite{DeligiannidisDoucetPitt2017}.

\section{Implementing correlated log-target estimators}
\label{app:correlated}
This section describes the implementations of the correlated importance sampling, sub-sampling and particle filtering algorithms.
These implementations are deterministic given the auxiliary variables $u$ given by the proposals in Section~3.1 in the main paper.

\begin{algorithm}[t]
    \caption{\textsf{Correlated importance sampling}}
    \footnotesize
	\textsc{Inputs:} $y$, $\theta$, $N$ and $u$. \\
	\textsc{Outputs:} $\log \widehat{\pi}^N(\theta | u)$ and $\widehat{\nabla}^N \log \pi(\theta | u)$. \\
	\textsc{Note:} all operations are carried out over $i = 1, \ldots, N$.
	\algrule[.4pt]
	\begin{algorithmic}[1]
		\FOR{$t=1$ to $T$}
			\STATE Simulate from the proposal by
			\begin{align*}
			x_t^{(i)} = \mu + \sigma^2 u^{(i)}_t,
			\end{align*} using random variables $u_{1:N, t}$.
			\STATE Calculate the weights by
			\begin{align*}
			w^{(i)}_t = \frac{W^{(i)}_t}{\sum_{j=1}^N W^{(j)}_t},
			\quad
			W^{(i)}_t = \mathcal{N} \left( y_t;x^{(i)}_t,\sigma_e^2 \right)
			\end{align*}
		\ENDFOR
		\STATE Estimate the log-target by \eqref{eq:app:llest}.
		\STATE Estimate the gradients of the log-target by \eqref{eq:app:random:effects:gradients}.
	\end{algorithmic}
	\label{alg:app:importance:sampling}
	\end{algorithm}

\subsection{Importance sampling for the random effects model}
Algorithm~3 in the main paper requires estimates of the log-posterior as well as its gradients.
The former can be computed by
\begin{align}
	\log \widehat{\pi}^N(\theta | y)
	&=
	\log p(\theta)
	+
	\sum_{t=1}^T
	\log
	\left[
	\sum_{i=1}^N
	w^{(i)}_t
	\right]
	- T \log N,
	\label{eq:app:llest}
\end{align}
where the weights $w^{(i)}_t$ are generated using the procedure outlined in Algorithm~\ref{alg:app:importance:sampling}.
In a similar manner, it is possible to compute the estimate of the gradients with respect to $\mu$ and $\tilde{\sigma}$ in the random effects model by
\begin{subequations}
\begin{align}
	\widehat{\nabla}^N_{\mu} \log \pi(\theta)
	&=
	\nabla \log p(\mu)
	+
	\widehat{\nabla}^N_{\mu} \ell(\theta), \\
	\widehat{\nabla}^N_{\tilde{\sigma}} \log \pi(\theta)
	&=
	\nabla \log p(\sigma)
	+
	\widehat{\nabla}^N_{\tilde{\sigma}} \ell(\theta), \\
	\widehat{\nabla}^N_{\mu} \ell(\theta)
	&=
	\sigma^2
	\sum_{i=1}^N
	\sum_{t=1}^T
	w^{(i)}_t
	\left(
		x^{(i)}_t
		-
		\mu
	\right), \\
	\widehat{\nabla}^N_{\tilde{\sigma}} \ell(\theta)
	&=
	\sum_{i=1}^N
	\sum_{t=1}^T
	w^{(i)}_t
	\left\{
	\sigma^{-2}
	\left(
		x^{(i)}_t
		-
		\mu
	\right)^2
	-
	1
	\right\},
\end{align}%
\label{eq:app:random:effects:gradients}%
\end{subequations}%
\noindent where $\ell(\theta) \triangleq \log p(y | \theta)$ denotes the log-likelihood function.
Note that the reparameterisation of $\tilde{\sigma} = \log(\sigma)$ is taken into account in the computation of the gradient.
Here, the particles $x^{(i)}_t$ are generated using the procedure outlined in Algorithm~\ref{alg:app:importance:sampling}.

\begin{algorithm}[t]
    \caption{\textsf{Correlated stratified sub-sampling}}
    \footnotesize
	\textsc{Inputs:} $T$, $N$ and $u$. \\
	\textsc{Outputs:} $\log \widehat{\pi}^N(\theta | u)$, $\widehat{\nabla}^N \log \pi(\theta | u)$ and $\widehat{\nabla}^{2,N} \log \pi(\theta | u)$.
	\algrule[.4pt]
	\begin{algorithmic}[1]
		\STATE Set $j = 1$ and $\mathcal{D} = \emptyset$.
		\FOR{$t=1$ to $T$}
			\STATE Set $\tilde{w}_t = t / T$.
		\ENDFOR
		\FOR{$i=1$ to $N$}
			\WHILE{$\tilde{w}_j < (u_i + i - 1)/N$ and $j < T$}
			\STATE Set $j = j + 1$.
			\ENDWHILE
			\STATE Set $\mathcal{D} = \mathcal{D} \cup \{y_j, X_i\}$.
		\ENDFOR
		\STATE Estimate the log-target by \eqref{eq:app:logistic:ll}.
		\STATE Estimate the gradients of the log-target by \eqref{eq:app:logistic:gradient}.
		\STATE Estimate the Hessian of the log-target by \eqref{eq:app:logistic:hessian}.
	\end{algorithmic}
	\label{alg:app:subsamp}
\end{algorithm}

\subsection{Sub-sampling for logistic regression model}
The log-target for the logistic regression model is given by
\begin{align}
	\log \pi(\theta)
	&=
	\log p(\theta)
	+
	\sum_{t=1}^T
	y_t \log p_t
	+
	(1 - y_t)
	\log( 1 - p_t),
	\label{eq:app:logistic:ll}
	\\
	p_t
	&=
	\left[ 1 + \exp(- \beta X^{\top}_t) \right]^{-1},
	\nonumber
\end{align}
where $\beta$ and $X_t$ denotes the model parameters and the regressors at time $t$, respectively.
The gradients of the log-target given a subset of the data $\mathcal{D}$ can be computed by
\begin{align}
	\nabla_{\beta_l} \log \pi(\theta)
	&=
	\nabla_{\beta_l} \log p(\theta) + \nabla_{\beta_l} \ell(\theta),
	\label{eq:app:logistic:gradient}
	\\
	\nabla_{\beta_l} \ell(\theta)
	&=
	\sum_{y_i, X_i \in \mathcal{D}}
	y_i
	G_i
	-
	(1 - y_i)
	G_i
	\exp(\beta X^{\top}_i)
	,
	\nonumber \\
	G_i
	&=
	x_{i, l} \left[ 1 + \exp(\beta X^{\top}_i) \right]^{-1}.
	\nonumber
\end{align}
In a similar manner, the Hessian of the log-target is given by
\begin{align}
	\nabla_{\beta_l}^2 \log \pi(\theta)
	&=
	\nabla^2_{\beta_l} \log p(\theta) + \nabla^2_{\beta_l} \ell(\theta),
	\label{eq:app:logistic:hessian}
	\\
	\nabla^2_{\beta_l} \ell(\theta)
	&=
	\sum_{y_i, X_i \in \mathcal{D}}
	\left[
	y_i H_{i, 1}
	+
	(1 - y_i) H_{i, 0}
	\right]
	X_i X_i^{\top}
	,
	\nonumber \\
	H_{i, 1}
	&=
	-
	\left[1 + \exp(\beta X^{\top}_i) \right]^{-2}
	\exp(\beta X^{\top}_i),
	\nonumber \\
	H_{i, 0}
	&=
	-
	\left[1 + \exp(-\beta X^{\top}_i) \right]^{-2}
	\exp(-\beta X^{\top}_i).
	\nonumber
\end{align}
The subset $\mathcal{D}$ is generated by Algorithm~\ref{alg:app:subsamp}.
Recall that second-order Gaussian proposal (15 in main paper) makes use of the negative inverse Hessian of the log-target.

\subsection{Particle filtering for stochastic volatility model}
We make use of a bootstrap particle filter (bPF; \citealp{DoucetJohansen2011}) to estimate the log-likelihood for the stochastic volatility model.
The gradient can then be estimated using a standard fixed-lag particle smoother (flPS; \citealp{OlssonCappeDoucMoulines2008}) using the output from the bPF.
The complete algorithm for flPS is available in e.g., \cite[Algorithm~2]{DahlinLindstenSchon2015a}.
The correlated version of bPF is given in Algorithm~\ref{alg:app:bpf}.

Hence, the log-likelihood is computed using the same expression \eqref{eq:app:llest} as for importance sampling but the particles $x^{(i)}_t$ are instead generated by sequential importance sampling with resampling.
Note that we are required to sort the particles after the propagation step, see the discussion about smooth particle filter by \cite{MalikPitt2011} for more details.

We make use of the probability transform to generate the uniform random variables required for the systematic resampling step.
Hence, $u$ is a $(N+1) \cdot (T+1)$-variate Gaussian random variable, where $u_{2:N+1,t}$ is used directly in the propagation step and $u_{1,t}$ is used in the resampling step after a transformation into a uniform random variable.

\begin{algorithm}[t]
    \caption{\textsf{Correlated bootstrap particle filter}}
    \footnotesize
	\textsc{Inputs:} $y$, $\theta$, and $u$. \\
	\textsc{Outputs:} $\log \widehat{\pi}^N(\theta|u)$. \\
	\textsc{Note:} all operations are carried out over $i,j = 1, \ldots, N$.
	\algrule[.4pt]
	\begin{algorithmic}[1]
		\STATE Sample the initial particles and set equal weights by
		\begin{align*}
			x^{(i)}_0
			=
			\mu
			+
			\frac{\sigma_v}{1 - \phi^2}
			u_{i+1, 1},
			\quad
			w_0^{(i)}=N^{-1}.
		\end{align*}
		\FOR{$t=1$ to $T$}
			\STATE Apply the CDF approach to transform $u_{1, t+1}$ into a uniform random number $\bar{u}$.
			\STATE Apply systematic resampling with $\bar{u}$ as the random variable to sample the ancestor indices $a^{(1:N)}_t$.
			\STATE Propagate the particles by computing
			\begin{align*}
				x_t^{(i)}
				&=
				\mu_{t}^{(i)}
				+
				\sigma_{t}^{(i)}
				u_{i+1, t+1},
				\\
				\mu_{t}^{(i)}
				&=
				\mu
				+
				\phi( x_{t-1}^{a^{(i)}_t} - \mu )
				+
				\sigma_v
				\rho
				\exp(-y_{t-1}/2),
				\\
				\sigma_t^{(i)}
				&=
				\sigma_v
				\sqrt{1 - \rho^2}.
			\end{align*}
			\STATE Extend the trajectory by $x_{0:t}^{(i)} = \Big\{ x_{0:t-1}^{a_t^{(i)}}, x_{t}^{(i)} \Big\}$ and sort them in ascending order according to $x_t^{(i)}$.
			\STATE Calculate the particle weights by
			\begin{align*}
				w^{(i)}_t
				=
				\frac{W^{(i)}_t}{\sum_{j=1}^N W^{(j)}_t},
				\quad
				W^{(i)}_t
				=
				\mathcal{N}(y_t; 0, \exp(x_t)).
			\end{align*}
		\ENDFOR
		\STATE Estimate $\widehat{p}_{\theta}(y;u)$ by \eqref{eq:app:llest}.
	\end{algorithmic}
	\label{alg:app:bpf}
\end{algorithm}

The gradient of the log-target can be computed using Fisher identity \citep{CappeMoulinesRyden2005},
\begin{align*}
	G(\theta)
	&=
	\nabla
	\log
	\pi(\theta)
	=
	\nabla
	\log
	p(\theta)
	+
	\nabla\,
	\mathbb{E}_{\theta}
	\left[
		\log p(x, y | \theta)
	\Big|
	y
	\right].
\end{align*}
This requires the gradients of the logarithm of the joint distribution of states and observations $\log p(x, y | \theta)$ given by
\begin{align*}
	\nabla
	\log p(x, y | \theta)
	=
	\sum_{t=1}^T
	\left[
		\nabla
		\log
		p(x_t | x_{t-1}, \theta)
		+
		\nabla
		\log
		p(y_t | x_t, \theta)
	\right],
\end{align*}
if the distribution of the initial state is independent of the parameters $\theta$.
For the stochastic volatility model, these expressions are given by
\begin{align*}
	\nabla_{\mu}
	\log p(x, y | \theta)
	&=
	\sum_{t=1}^T
	Q
	G_t
	(1 - \phi), \\
	\nabla_{\tilde{\phi}}
	\log p(x, y | \theta)
	&=
	\sum_{t=1}^T
	Q
	G_t
	(1 - \phi^2), \\
	\nabla_{\tilde{\sigma}_v}
	\log p(x, y | \theta)
	&=
	\sum_{t=1}^T
	\Big\{
	Q
	G_t
	\Big[
		G_t
	+
	\sigma_v \rho \exp(-0.5 x_{t-1}) y_{t-1}
	\Big]
	-
	1 \Big\},
	\\
	\nabla_{\tilde{\rho}}
	\log p(x, y | \theta)
	&=
	\sum_{t=1}^T
	\Big\{
	\rho - Q \rho G_t^2 +
	\sigma_v^{-1} G_t \exp(-0.5 x_{t-1}) y_{t-1} \Big\},
\end{align*}
where the following notation has been introduced for brevity
\begin{align*}
	G_t
	&=
	x_t
	-
	\mu
	-
	\phi(x_{t-1} - \mu)
	-
	\rho \sigma_v \exp(-0.5 x_{t-1}) y_{t-1}, \\
	Q
	&=
	\sigma_v^{-2}
	(1 - \rho^2)^{-1}.
\end{align*}
Note that we have taken into account the reparameterisation of the model when computing the gradients.
Recall that $\tilde{\phi} = \text{atanh}(\phi)$, $\tilde{\sigma}_v = \log(\sigma_v)$ and $\tilde{\rho} = \log(\rho)$.

The Hessian can be computed using the smoother and the Louis identity \citep[p.\ 352]{CappeMoulinesRyden2005} in a similar manner.
However, the second derivatives are required for the implementation, which are a bit more involved to compute.
Please have a look at the source code for the details.

\clearpage
\section{Additional results}
This section contains some additional results for the examples in the main paper.

\begin{figure}[p]
	\centering
	\includegraphics[width=\textwidth]{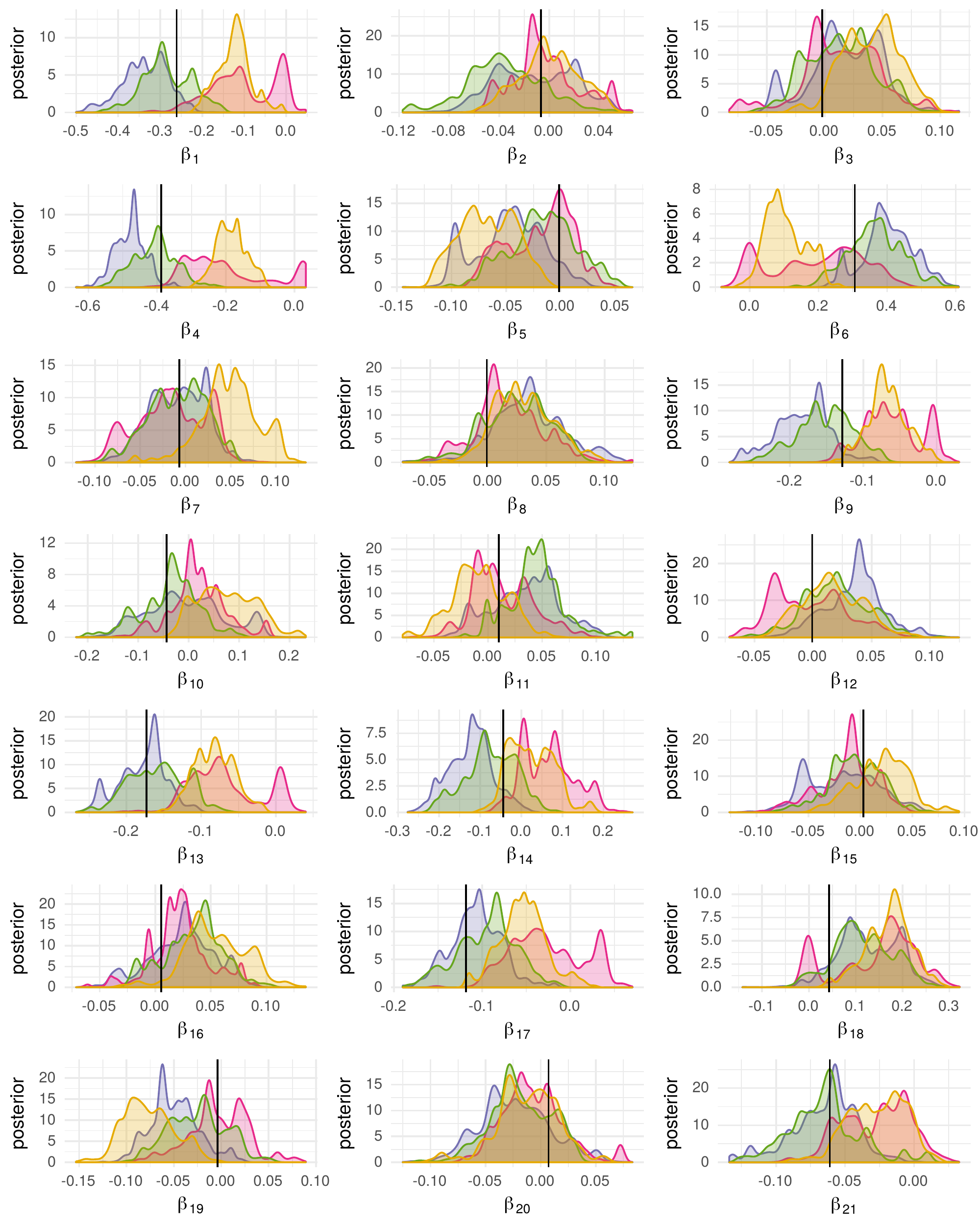}
	\caption{\footnotesize The posterior estimates for $\beta_{1:21}$ obtained by pmMH2 (purple) and Algorithm~3 in the main paper using BFSG (magenta), LS (green) and SR1 (yellow) in the logistic regression model using the Higgs data.}
	\label{fig:example2b-higgs-supp-post}
\end{figure}

\begin{figure}[p]
	\centering
	\includegraphics[width=\textwidth]{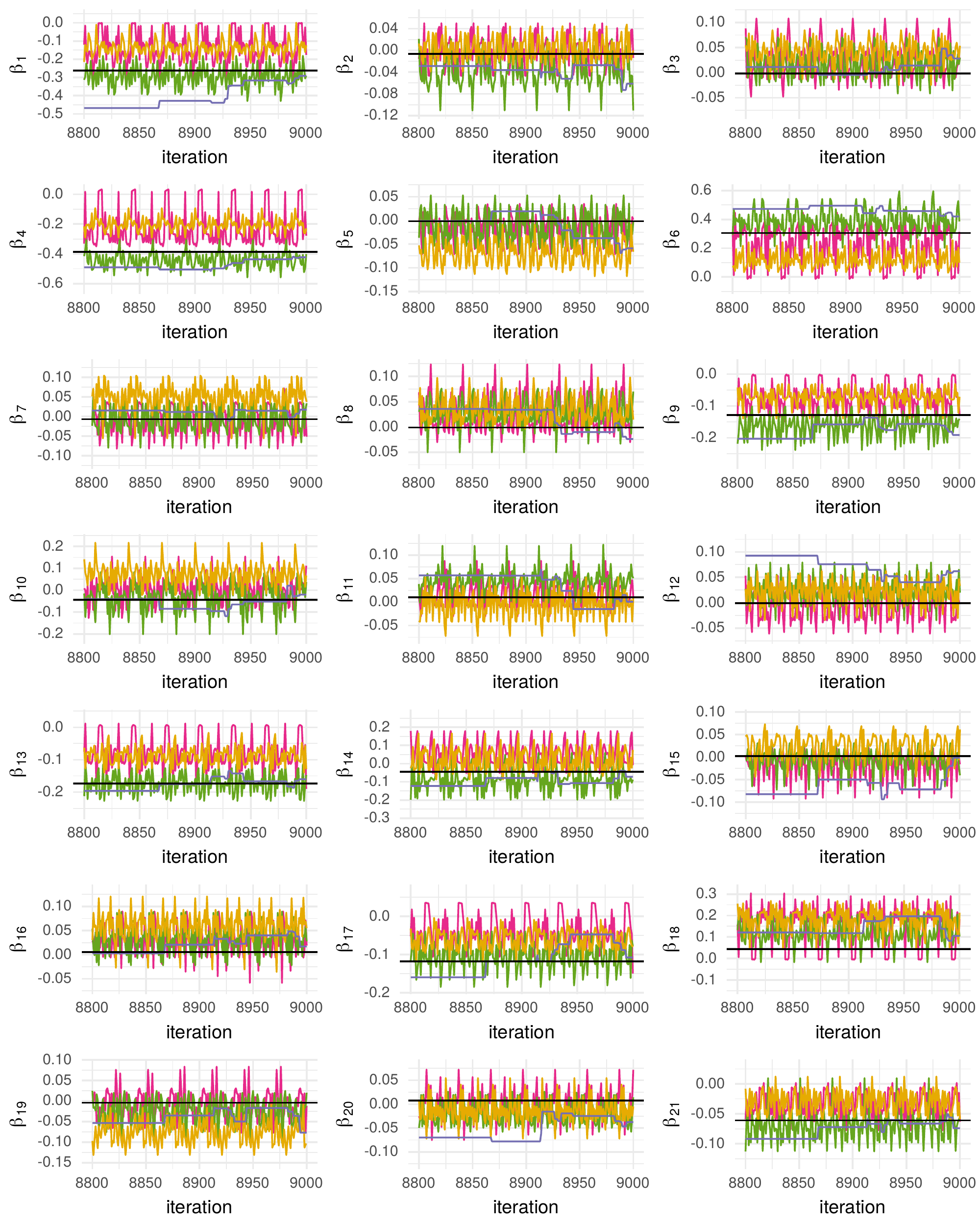}
	\caption{\footnotesize Trace plots for $\beta_{1:21}$ obtained by pmMH2 (purple) and Algorithm~3 in the main paper using BFSG (magenta), LS (green) and SR1 (yellow) in the logistic regression model using the Higgs data.}
	\label{fig:example2b-higgs-supp-trace}
\end{figure}

\begin{figure}[p]
	\centering
	\includegraphics[width=\textwidth]{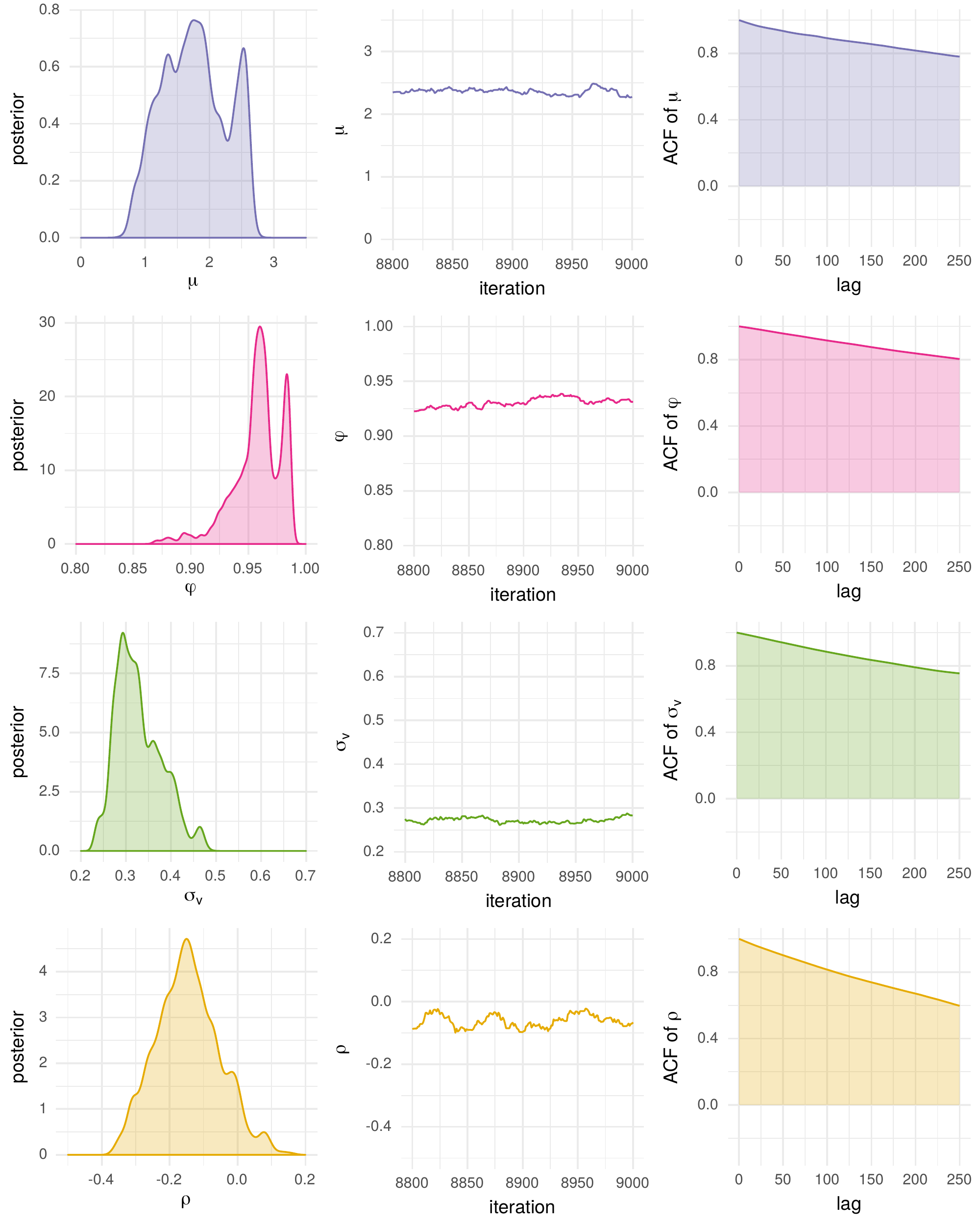}
	\caption{\footnotesize Posterior estimates from 15,000 samples (left), trace plots (middle) and empirical ACF (right) from pmMH2 for $\mu$ (purple), $\phi$ (magenta), $\sigma_v$ (green) and $\rho$ (yellow) in the stochastic volatility model using the Bitcoin data.}
	\label{fig:example3-stochastic-volatility-supp-mh2}
\end{figure}

\begin{figure}[p]
	\centering
	\includegraphics[width=\textwidth]{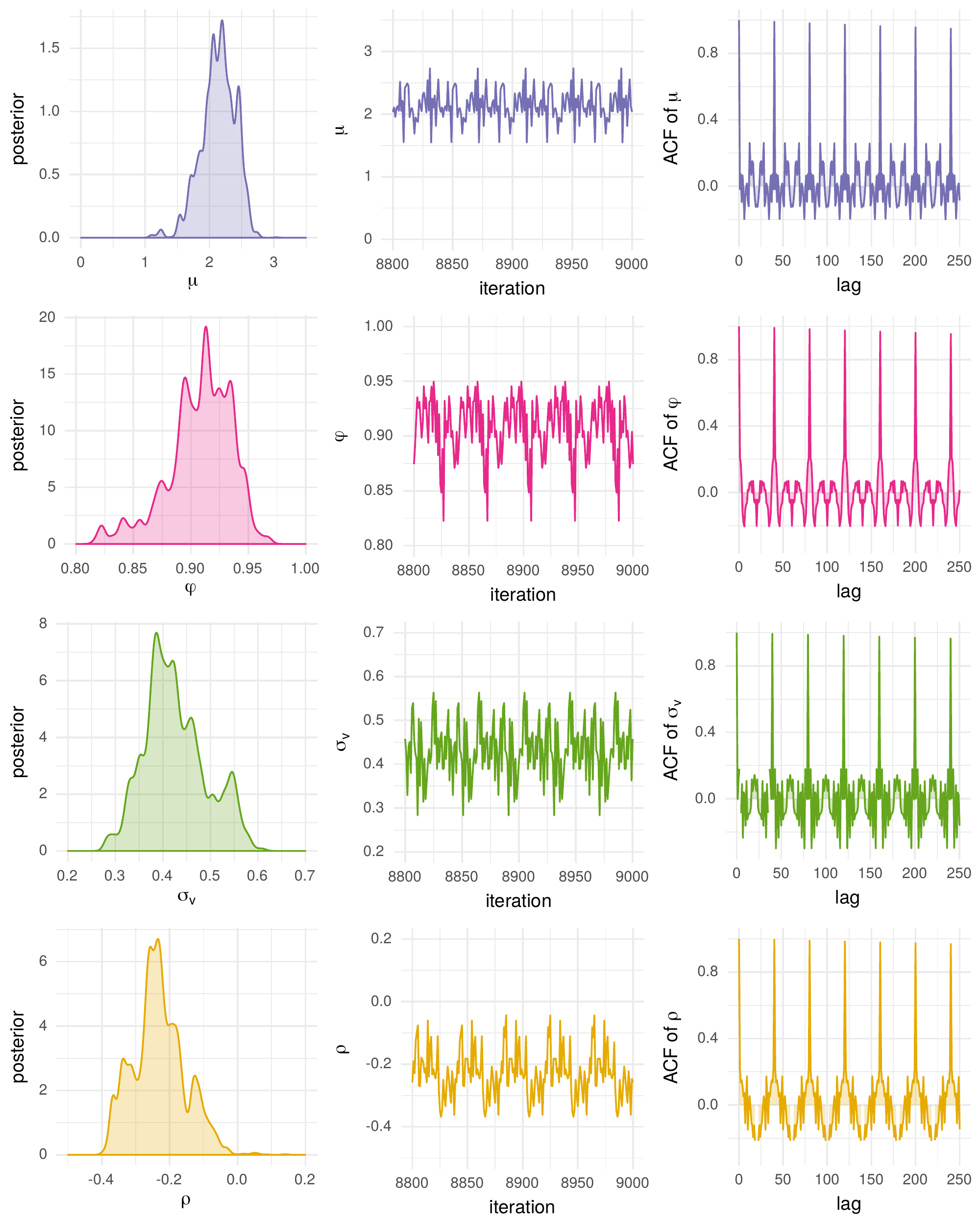}
	\caption{\footnotesize Posterior estimates from 15,000 samples (left), trace plots (middle) and empirical ACF (right) from Algorithm~3 in the main paper using BFGS updates for $\mu$ (purple), $\phi$ (magenta), $\sigma_v$ (green) and $\rho$ (yellow) in the stochastic volatility model using the Bitcoin data.}
	\label{fig:example3-stochastic-volatility-supp-bfgs}
\end{figure}

\begin{figure}[p]
	\centering
	\includegraphics[width=\textwidth]{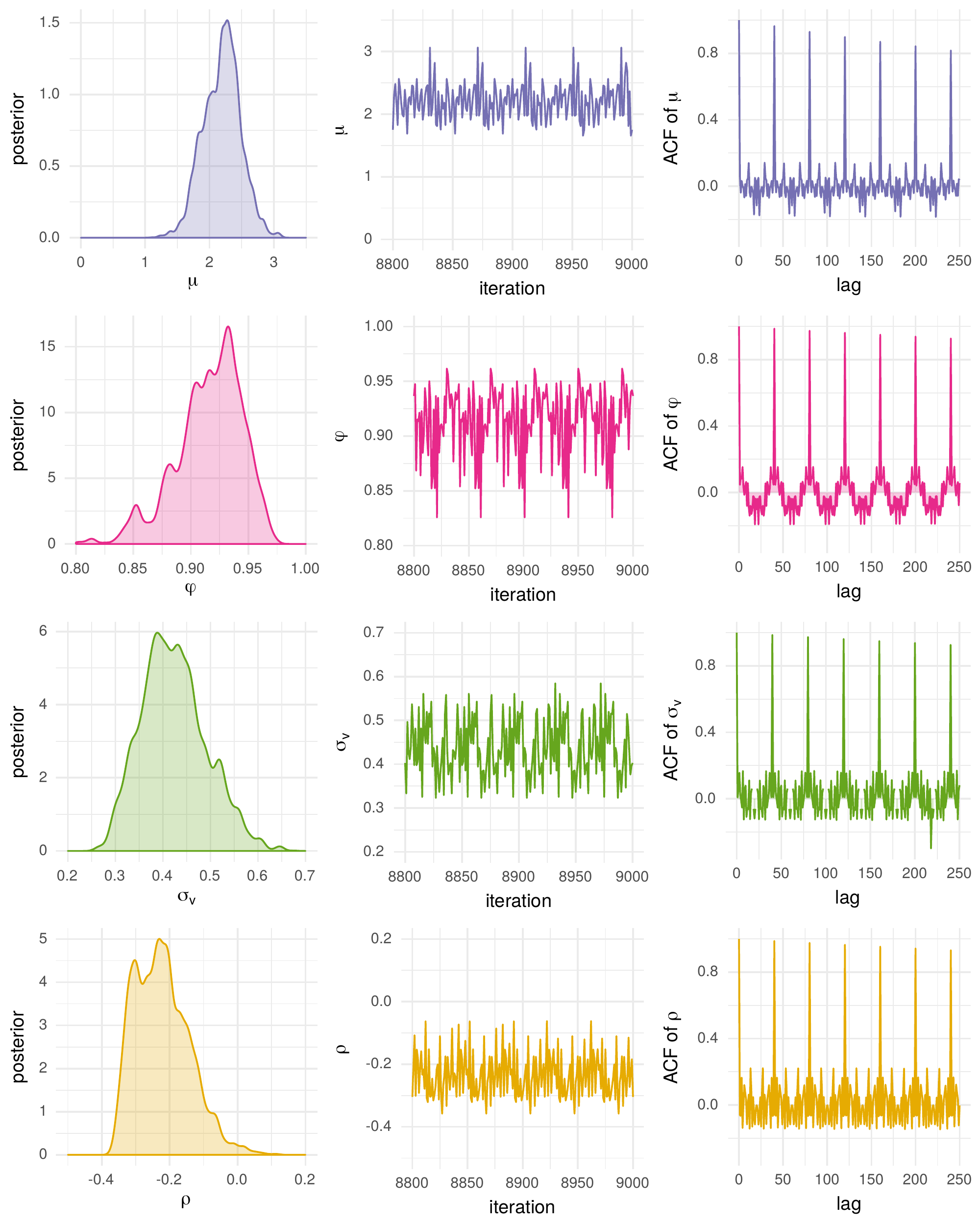}
	\caption{\footnotesize Posterior estimates from 15,000 samples (left), trace plots (middle) and empirical ACF (right)  from Algorithm~3 in the main paper using LS updates for $\mu$ (purple), $\phi$ (magenta), $\sigma_v$ (green) and $\rho$ (yellow) in the stochastic volatility model using the Bitcoin data.}
	\label{fig:example3-stochastic-volatility-supp-ls}
\end{figure}

\begin{figure}[p]
	\centering
	\includegraphics[width=\textwidth]{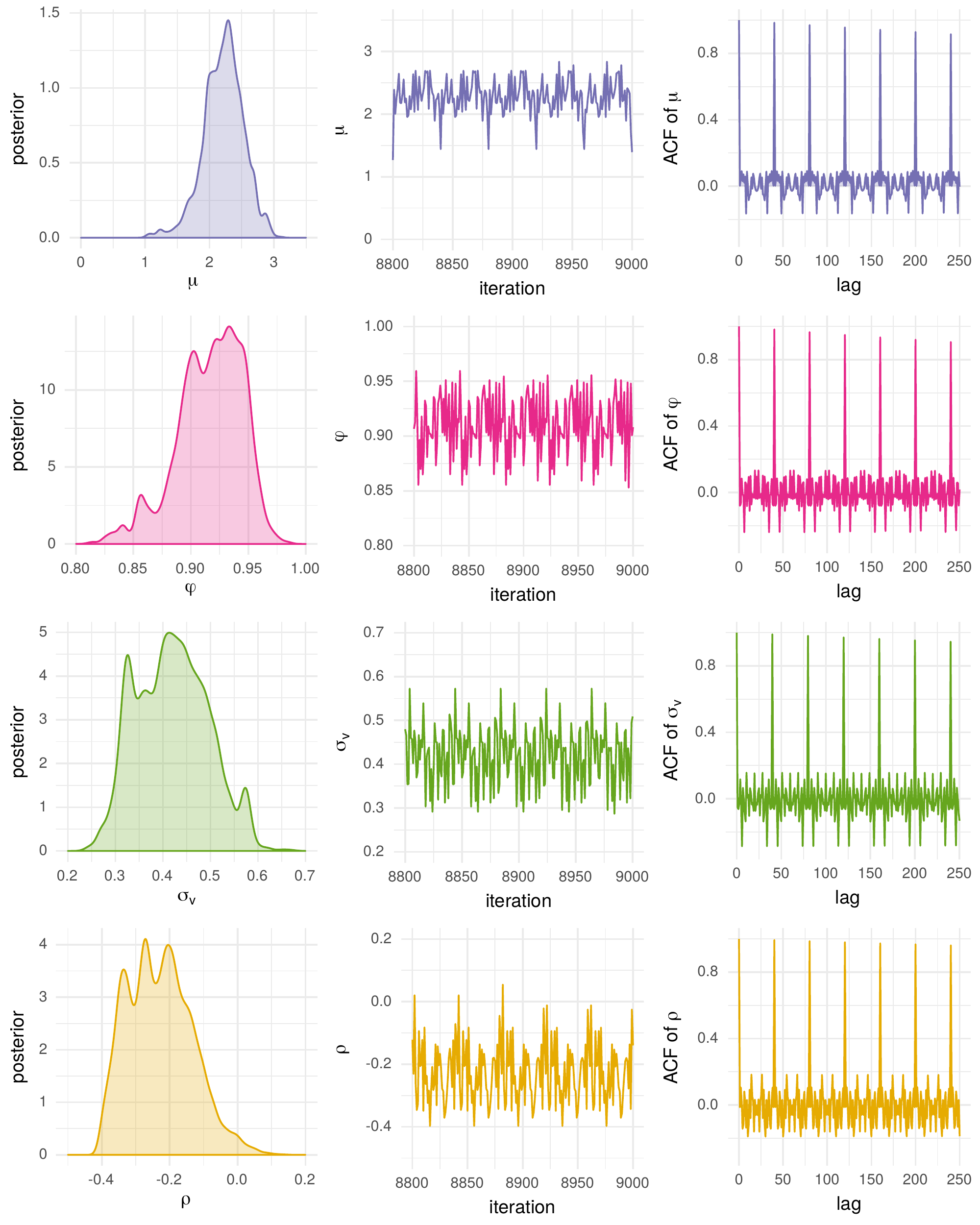}
	\caption{\footnotesize Posterior estimates from 15,000 samples (left), trace plots (middle) and empirical ACF (right)  from Algorithm~3 in the main paper using SR1 updates for $\mu$ (purple), $\phi$ (magenta), $\sigma_v$ (green) and $\rho$ (yellow) in the stochastic volatility model using the Bitcoin data.}
	\label{fig:example3-stochastic-volatility-supp-sr1}
\end{figure}

\end{document}